\documentclass[11pt]{article} %for abst
\usepackage{fullpage}
\usepackage{subfigure}
\usepackage{graphicx}

\newcommand{\nop}[1]{}
\newcommand{\TMM}{\mbox{$T_{\it mM}$}}

\newtheorem{theorem}{Theorem}
\newtheorem{corollary}{Corollary}
\newtheorem{lemma}{Lemma}

\newtheorem{example}{Example}
\newenvironment{proof}{{\noindent\bf\em Proof.}}{\qed}

\newcommand{\margin}[1]{\marginpar{\tiny #1}}

\makeatletter
\def\cramped {\parskip\@outerparskip\@topsep\parskip
  \@topsepadd2pt\itemsep0pt}
\makeatletter

\begin{document}
\def\lcp{{\mbox{\it lcp}}}
\def\lca{{\mbox{\it lca}}}
\def\deepest{{\mbox{\it deepest\_node}}}
\def\height{{\mbox{\it height}}}
\def\LA{{\mbox{\it LA}}}
\def\depth{{\mbox{\it depth}}}
\def\wdepth{{\mbox{\it wdepth}}}
\def\leaf{{\mbox{\it leaf}}}
\def\parent{{\mbox{\it parent}}}
\def\enclose{{\mbox{\it enclose}}}
\def\alph{\mathcal{A}}
\def\sizealph{\sigma}
\def\argmin{\mathop{\rm argmin}}
\def\argmax{\mathop{\rm argmax}}
\def\mathdef{\stackrel{\rm def}{=}}
\def\rank{{\mbox{\it rank}}}
\def\select{{\mbox{\it select}}}
\def\findclose{{\mbox{\it findclose}}}
\def\findopen{{\mbox{\it findopen}}}
\def\LA{{\mbox{\it level\_ancestor}}}
\def\rmq{{\mbox{\it rmq}}}
\def\RMQ{{\mbox{\it RMQ}}}
\def\rmqi{{\mbox{\it rmqi}}}
\def\RMQI{{\mbox{\it RMQi}}}
\def\degree{{\mbox{\it degree}}}
\def\child{{\mbox{\it child}}}
\def\childrank{{\mbox{\it child\_rank}}}
\def\excess{{\mbox{\it excess}}}
\def\sumw{{\mbox{\it sum}}}
\def\fwd{{\mbox{\it fwd\_search}}}
\def\bwd{{\mbox{\it bwd\_search}}}
\newcommand{\FID}{\textsc{fid}}
\newcommand{\OP}{{\fbox{\tt (}}}
\newcommand{\CP}{{\fbox{\tt )}}}
\newcommand{\CCP}{{\fbox{\tt ))}}}
\newcommand{\COP}{{\fbox{\tt )(}}}
\newcommand{\OCP}{{\fbox{\tt ()}}}
\newcommand{\leafrank}{{\it leaf\_rank}}
\newcommand{\leafselect}{{\it leaf\_select}}
\newcommand{\preorderrank}{{\it pre\_rank}}
\newcommand{\preorderselect}{{\it pre\_select}}
\newcommand{\postorderrank}{{\it post\_rank}}
\newcommand{\postorderselect}{{\it post\_select}}
\newcommand{\inorderrank}{{\it in\_rank}}
\newcommand{\inorderselect}{{\it in\_select}}
\newcommand{\leftmostleaf}{{\it lmost\_leaf}}
\newcommand{\rightmostleaf}{{\it rmost\_leaf}}
\newcommand{\ins}{{\it insert}}
\newcommand{\del}{{\it delete}}
\newcommand{\attach}{{\it attach}}
\newcommand{\detach}{{\it detach}}
\def\rankopen{{\mbox{$\it rank_\OP$}}}
\def\rankclose{{\mbox{$\it rank_\CP$}}}
\def\selectopen{{\mbox{$\it select_\OP$}}}
\def\selectclose{{\mbox{$\it select_\CP$}}}
\def\nextsibling{{\mbox{\it next\_sibling}}}
\def\prevsibling{{\mbox{\it prev\_sibling}}}
\def\firstchild{{\mbox{\it first\_child}}}
\def\lastchild{{\mbox{\it last\_child}}}
\def\subtreesize{{\mbox{\it subtree\_size}}}
\def\deepestnode{{\mbox{\it deepest\_node}}}

\def\levelnext{{\mbox{\it level\_next}}}
\def\levelprev{{\mbox{\it level\_prev}}}
\def\levelleftmost{{\mbox{\it level\_lmost}}}
\def\levelrightmost{{\mbox{\it level\_rmost}}}

\newcommand{\inspect}{{\it inspect}}
\newcommand{\isleaf}{{\it isleaf}}
\newcommand{\isancestor}{{\it isancestor}}

\def\polylog{{\mathop{\mathrm{polylog}}\nolimits}}
\def\poly{{\mathop{\mathrm{poly}}\nolimits}}
\def\lg{{\mathop{\mathrm{lg}}\nolimits}}
\newtheorem{definition}{Definition}
\newtheorem{problem}{Problem}
\def\square{\hbox{\rlap{$\sqcap$}$\sqcup$}}
\newcommand{\qed}{\hspace*{\fill}$\square$\vspace{1.5ex}}

\title{Fully-Functional Static and Dynamic Succinct Trees
\thanks{A preliminary version of this paper appeared in {\em Proc. SODA 2010},
pp. 134--149.}}
\author{%
Gonzalo Navarro\thanks{%
Department of Computer Science, University of Chile.
{\tt gnavarro@dcc.uchile.cl}.
Funded in part by Millennium Institute for Cell Dynamics and
Biotechnology (ICDB), Grant ICM P05-001-F, Mideplan, Chile.}
\and
Kunihiko Sadakane\thanks{%
Principles of Informatics Research Division, National Institute of
Informatics,
2-1-2 Hitotsubashi, Chiyoda-ku, Tokyo 101-8430, Japan.
{\tt sada@nii.ac.jp}.
Work supported in part by the Grant-in-Aid
of the Ministry of Education, Science, Sports and Culture of Japan.}
}	
\date{}

\maketitle

\newcommand{\keyw}[1]{{\bf #1}}
\newcommand{\Order}{\mathcal{O}}
\newcommand{\order}{o}

\begin{abstract}

We propose new succinct representations of ordinal trees,
which have been studied extensively.  It is known that any $n$-node
static tree can be represented in $2n + \order(n)$ bits and a number of
operations on the tree can be supported in constant time under the
word-RAM model.  However the data structures are complicated and 
difficult to dynamize.
We propose a simple and flexible data structure, called the range min-max
tree, that reduces the large number of relevant tree operations considered in
the literature to a few primitives that are carried out in constant time on 
sufficiently small trees. The result is extended to trees of arbitrary size,
achieving $2n + \Order(n /\polylog(n))$ bits of space, which is optimal for 
some operations.
The redundancy is significantly lower than any previous proposal.
For the dynamic case, where insertion/deletion of nodes is allowed, 
the existing data structures support very limited operations.  
Our data structure builds on the range min-max tree to achieve
$2n+\Order(n/\log n)$ bits of space and $\Order(\log n)$ time for all the
operations. We also propose an improved data structure using 
$2n+\Order(n\log\log n/\log n)$ bits and improving the time to the optimal
$\Order(\log n/\log \log n)$ for most operations. We extend our support to
forests, where whole subtrees can be attached to or detached from others, in
time $\Order(\log^{1+\epsilon} n)$ for any $\epsilon>0$.

Our techniques are of independent interest. An immediate derivation gives
improved solution to range minimum/maximum queries where consecutive elements
differ by $\pm 1$, achieving $\Order(n+n/\polylog(n))$ bits of space. A second
one stores an array of numbers supporting operations $sum$ and $search$ and
limited updates, in optimal time $\Order(\log n /\log\log n)$. A third one
allows representing dynamic bitmaps and sequences supporting
rank/select and indels, within zero-order entropy bounds and optimal time
$\Order(\log n / \log\log n)$ for all operations on bitmaps and polylog-sized
alphabets, and $\Order(\log n \log \sigma / (\log\log n)^2)$ on larger 
alphabet sizes $\sigma$. This improves upon the best existing bounds for 
entropy-bounded storage of dynamic sequences, compressed full-text self-indexes,
and compressed-space construction of the Burrows-Wheeler transform.

\end{abstract}

\section{Introduction} 

Trees are one of the most fundamental data structures, needless to say.
A classical representation of a tree with $n$ nodes uses $\Order(n)$ pointers 
or words. Because each pointer must distinguish all the nodes, it requires
$\log n$ bits\footnote{The base of logarithm is $2$ throughout this paper.}
in the worst case. Therefore the tree occupies $\Theta(n \log n)$ bits.
This causes a space problem for storing a large set of items in a tree.
Much research has been devoted to reducing the space to represent static 
trees~\cite{Jacobson89,MunRam01,MunRamRao01,MunRao04,GearyRRR04,GRR04,BenDemMunRamRamRao05,FLMM05,ChiLinLu05,DRR06,LuYeh07,HMR07,BMHR07,GGGRR07,Sada07a,JanSadSun07,FarMun08}
and dynamic trees~\cite{MunroRamanStormSODA01,RamRao03,CHLS07,Arr08},
achieving so-called \emph{succinct data structures} for trees.

A succinct data structure stores objects
using space close to the information-theoretic
lower bound, while simultaneously supporting a number of primitive operations
on the objects in constant time.
Here the information-theoretic lower bound for storing an object from a
universe with cardinality $L$ is $\log L$
bits because in the worst case
this number of bits is necessary to distinguish any two objects.

In this paper we are interested in \emph{ordinal trees}, in which
the children of a node are ordered.  
The information-theoretic lower bound for representing an ordinal tree with 
$n$ nodes is $2n-\Theta(\log n)$ bits because there exist
${{2n-1}\choose{n-1}}/(2n-1) = 2^{2n}/\Theta(n^\frac{3}{2})$ 
such trees~\cite{MunRam01}.
The size of a succinct data structure storing an object from the universe
is typically $(1+\order(1))\log L$ bits.
We assume that the computation model is
the word RAM with word length $\Theta(\log n)$
in which arithmetic and logical operations on $\Theta(\log n)$-bit integers
and $\Theta(\log n)$-bit memory accesses can be done in constant time.

Basically there exist three types of succinct representations
of ordinal trees: the balanced parentheses sequence (BP)~\cite{Jacobson89,MunRam01},
the level-order unary degree sequence (LOUDS)~\cite{Jacobson89,DRR06},
and the depth-first unary degree sequence (DFUDS)~\cite{BenDemMunRamRamRao05,JanSadSun07}.
An example of them is shown in Figure~\ref{fig:tree}.
LOUDS is a simple representation, but it lacks many basic operations, such as
the subtree size of a given node. Both BP and DFUDS build on a sequence of 
balanced parentheses, the former using the intuitive depth-first-search 
representation and the latter using a more sophisticated one. The advantage of
DFUDS is that it supports a more complete set of operations by simple 
primitives, most notably going to the $i$-th child of a node in constant time.
In this paper we focus on 
the BP representation, and achieve constant time for a large set of
operations, including all those handled with DFUDS. Moreover, as we manipulate 
a sequence of balanced parentheses, our data structure can be used to
implement a DFUDS representation as well.

\subsection{Our contributions}

We propose new succinct data structures for ordinal trees encoded with
balanced parentheses, in both static and dynamic scenarios.
\paragraph{\em Static succinct trees.} 
For the static case we obtain the following result.

\begin{theorem}\label{th:main}
For any ordinal tree with $n$ nodes, all operations in
Table~\ref{tab:ops} except $\ins$ and $\del$
are carried out in constant time $\Order(c)$
with a data structure using $2n + \Order(n/\log^c n)$ bits of space
on a $\Theta(\log n)$-bit word RAM, for any constant $c > 0$.
The data structure can be constructed from the balanced parentheses sequence 
of the tree, in $\Order(n)$ time using $\Order(n)$ bits of space.
\end{theorem}

The space complexity of our data structures significantly improves upon the 
lower-order term achieved in previous representations.
For example, the extra data structure for {\LA} requires
$\Order(n \log\log n/\sqrt{\log n})$ bits~\cite{MunRao04},
or $\Order(n (\log\log n)^2/\log n)$ bits\footnote{This data structure is for DFUDS, but the same technique can be also applied to BP.}~\cite{JanSadSun07},
and that for {\child} requires $\Order(n/(\log\log n)^2)$ bits~\cite{LuYeh07}.
Ours requires $\Order(n /\log^c n)$ bits for all of the operations. We
show in the Conclusions that this redundancy is optimal for some operations.

The simplicity and space-efficiency of our data structures
stem from the fact that any query operation in 
Table~\ref{tab:ops} is reduced to a few basic operations
on a bit vector, which can be efficiently solved
by a \emph{range min-max tree}.
This approach is different from previous studies in which
each operation needs distinct auxiliary data structures.
Therefore their total space is the summation of all the data structures.
For example, the first succinct representation of BP~\cite{MunRam01}
supported only {\findclose}, {\findopen}, and {\enclose} (and other easy
operations) and each operation used different data structures.
Later, many further operations such as {\leftmostleaf}~\cite{MunRamRao01},
{\lca}~\cite{Sada07a},
{\degree}~\cite{ChiLinLu05}, {\child} and {\childrank}~\cite{LuYeh07},
{\LA}~\cite{MunRao04}, 
were added to this representation
by using other types of data structures for each.
There exists another elegant data structure for BP supporting
{\findclose}, {\findopen}, and {\enclose}~\cite{GearyRRR04}.
This reduces the size of the data structure for these basic operations,
but still has to add
extra auxiliary data structures for other operations.

\paragraph{\em Dynamic succinct trees.} 
Our approach is suitable for the dynamic maintenance of trees.
Former approaches in the static case use two-level data structures
to reduce the size, which causes difficulties in the dynamic case.
On the other hand, our approach using the range min-max tree is
easily applied in this scenario, resulting in simple and efficient
dynamic data structures. This is illustrated by the fact that all the
operations are supported. The following theorem summarizes our results.

\begin{theorem}\label{th:dyn}
On a $\Theta(\log n)$-bit word RAM, all operations on a dynamic ordinal tree
with $n$ nodes can be carried out within the worst-case complexities given in 
Table~\ref{tab:ops}, using a data structure that requires 
$2n + \Order(n \log\log n/\log n)$ bits. 
Alternatively, the operations of the table can be carried out in 
$\Order(\log n)$ time using $2n + \Order(n/\log n)$ bits of space.
\end{theorem}

Note we achieve time complexity $\Order(\log n / \log\log n)$ for most 
operations, including {\ins} and {\del}, if we solve {\degree}, {\child}, 
and {\childrank} naively. Otherwise we can achieve 
$\Order(\log n)$ complexity for these, yet also for {\ins} and {\del}.%
\footnote{In the conference version of this paper \cite{SNsoda10} we 
erroneously affirm we can obtain $\Order(\log n / \log\log n)$ for all these
operations, as well as $\LA$, $\levelnext/\levelprev$, and 
$\levelleftmost/\levelrightmost$, for which we can actually obtain only 
$\Order(\log n)$.}
The time complexity $\Order(\log n / \log\log n)$ is optimal: Chan et
al.~\cite[Thm.~5.2]{CHLS07} showed that just supporting the most basic 
operations of Table~\ref{tab:ops} ({\findopen}, {\findclose}, and {\em enclose},
as we will see) plus {\ins} and {\del}, requires this time even in the 
amortized sense, by a reduction from Fredman and Saks's lower bounds on 
{\rank} queries \cite{FreSak89}.

Moreover, we are able to attach and detach whole subtrees, in time 
$\Order(\log^{1+\epsilon} n)$ for any constant $\epsilon>0$ 
(see Section~\ref{sec:introdyn} for the precise details). 
These operations had never been considered before in 
succinct tree representations.

\paragraph{\em Byproducts.}
Our techniques are of more general interest. A subset of our data structure is
able to solve the well-known ``range minimum query'' problem \cite{BenFar00}.
In the important case where consecutive elements differ by $\pm 1$, we improve 
upon the best current space redundancy of $\Order(n\log\log n/\log n)$ bits
\cite{Fis10}.

\begin{corollary} \label{cor:rmq}
Let $E[0,n-1]$ be an array of numbers with the property that $E[i] - E[i-1]
\in \{-1,+1\}$ for $0 < i < n$, encoded as a bit vector $P[0,n-1]$ such that
$P[i] = 1$ if $E[i]-E[i-1] = +1$ and $P[i]=0$ otherwise. Then, in a RAM machine
we can preprocess $P$ in $\Order(n)$ time and $\Order(n)$ bits such that 
range maximum/minimum queries are answered in constant $\Order(c)$ time and 
$\Order(n/\log^c n)$ extra bits on top of $P$.
\end{corollary}

Another direct application, to the representation of a dynamic array of
numbers, yields an improvement to the best current alternative 
\cite{MN08} by a $\Theta(\log\log n)$ time factor. If the updates are limited, 
further operations $sum$ (that gives the sum of the numbers up to some
position) and $search$ (that finds the position where a given sum is exceeded) 
can be supported, and our complexity matches the lower bounds for {\em 
searchable partial sums} by P{\v a}tra{\c s}cu and Demaine \cite{PD06} 
(if the updates are not limited one can still use previous 
results \cite{MN08}, which are optimal in that general case).
% Precisely, they showed
%t_q ( log (log U/ log delta) + log (t_u / t_q)) = Omega(log n)
%where U is the upper bound of n and delta is the maximum value to
%add an entry.  If delta is polylog(n), we obtain
%t_q log log n = Omega(log n).  So t_q = Omega(log n/log log n).
We present our result in a slightly more general form.

\begin{lemma} \label{lem:dynpartialsums}
A sequence of $n$ variable-length constant-time self-delimiting%
\footnote{This means that one can distinguish the first code $x_i$ from a
bit stream $x_i\alpha$ in constant time.}
bit codes $x_1 \ldots x_n$, where $|x_i| = \Order(\log n)$,
can be stored within $(\sum |x_i|)(1+o(1))$ 
bits of space, so that we can $(i)$ compute any sequence of codes $x_i,
\ldots, x_j$, $(ii)$ update any code $x_i \leftarrow y$, $(iii)$ insert a new 
code $z$ between any pair of codes, and $(iv)$ delete any code $x_d$ from the 
sequence, all in $\Order(\log n / \log\log n)$ time (plus $j-i$ for
$(i)$). Moreover, let $f(x_i)$ be a nonnegative
integer function computable in constant time from the codes. If the updates
and indels are such that $|f(y)-f(x_i)|,f(z),f(x_d) = \Order(\log n)$, then we 
can also support operations $sum(i) = \sum_{j=1}^i f(x_i)$ and 
$search(s) = \max\{i, sum(i) \le s\}$ within the same time.
\end{lemma}

For example we can store $n$ numbers $0 \le a_i < 2^k$ within $kn+o(kn)$ bits, 
by using their $k$-bit binary representation $[a_i]_2$ as the code, and their
numeric value as $f([a_i]_2)=a_i$, so that we support $sum$ and $search$ on the
sequence of numbers. If the numbers are very different in magnitude we can 
$\delta$-encode them to achieve $(\sum \log a_i)(1+o(1)) + \Order(n)$ bits of space.
We can also store bits, seen as 1-bit codes, in $n+o(n)$ bits and and carry out
$sum = \rank$ and $search = \select$, insertions and deletions, in 
$\Order(\log n / \log\log n)$ time. 
% constant time decoding?
%We can even store a sequence of Huffman
%codes allowing updates to the sequence and direct access to any part of it.

A further application of our results to the compressed representation of 
sequences achieves a result summarized in the next theorem.

\begin{theorem} \label{thm:seqs}
Any sequence $S[0,n-1]$ over alphabet $[1,\sigma]$ can be stored in 
$nH_0(S)+\Order(n \log\sigma/\log^\epsilon n + \sigma\log^\epsilon n)$
bits of space, for any constant
$0<\epsilon<1$, and support the operations $\rank$, $\select$, $\ins$,
and $\del$, all in time $\Order\left(\frac{\log n}{\log\log n}
   \left(1+\frac{\log\sigma}{\log\log n}\right)\right)$. 
For polylogarithmic-sized alphabets, this is the optimal 
$\Order(\log n / \log\log n)$; otherwise it
is $O\left(\frac{\log n \log\sigma}{(\log\log n)^2}\right)$.
\end{theorem}

This time complexity slashes the the best current result \cite{GN08} by a
$\Theta(\log\log n)$ factor. The optimality of the polylogarithmic case stems 
again from Fredman and Saks' lower bound on {\rank} on dynamic bitmaps 
\cite{FreSak89}. This result has immediate applications to building compressed
indexes for text, building the Burrows-Wheeler transform within compressed
space, and so on.

\begin{table*}[tbp]
  \small
  \caption{Operations supported by our data structure. The time complexities
are for the dynamic case; in the static case all operations are performed in 
constant time.%
The first group is composed of basic operations, used to implement the others, 
but which could have other uses.}
  \label{tab:ops}

\bigskip
\footnotesize

  \begin{tabular}{lll|l}
  operation & description & \multicolumn{2}{c}{time complexity} \\ 
            &             &  \multicolumn{2}{c}{variant 1 $~~|~~$ variant 2} \\ \hline
  $\inspect(i)$ & $P[i]$ & 
			\multicolumn{2}{c}{$\Order(\log n/\log\log n)$} \\
  $\findclose(i)/\findopen(i)$ & position of parenthesis matching $P[i]$ & 
			\multicolumn{2}{c}{$\Order(\log n/\log\log n)$} \\
  $\enclose(i)$ & position of tightest open parent. enclosing $i$ & 
			\multicolumn{2}{c}{$\Order(\log n/\log\log n)$} \\
  $\rankopen(i)/\rankclose(i)$ & number of open/close parentheses in $P[0,i]$ & 
			\multicolumn{2}{c}{$\Order(\log n/\log\log n)$} \\
  $\selectopen(i)$/$\selectclose(i)$ & position of $i$-th open/close parenthesis & 
			\multicolumn{2}{c}{$\Order(\log n/\log\log n)$} \\
  $\rmqi(i,j)/\RMQI(i,j)$ & position of min/max excess value in range $[i,j]$ & 
			\multicolumn{2}{c}{$\Order(\log n/\log\log n)$} \\
\hline
  $\preorderrank(i)/\postorderrank(i)$ & preorder/postorder rank of node $i$ & 
			\multicolumn{2}{c}{$\Order(\log n/\log\log n)$} \\
  $\preorderselect(i)/\postorderselect(i)$ & the node with preorder/postorder
$i$ & 
			\multicolumn{2}{c}{$\Order(\log n/\log\log n)$}  \\
  $\isleaf(i)$ & whether $P[i]$ is a leaf & 
			\multicolumn{2}{c}{$\Order(\log n/\log\log n)$} \\
  $\isancestor(i,j)$ & whether $i$ is an ancestor of $j$ & 
			\multicolumn{2}{c}{$\Order(\log n/\log\log n)$} \\
  $\depth(i)$ & depth of node $i$ & 
			\multicolumn{2}{c}{$\Order(\log n/\log\log n)$} \\
  $\parent(i)$ & parent of node $i$ & 
			\multicolumn{2}{c}{$\Order(\log n/\log\log n)$} \\
  $\firstchild(i)/\lastchild(i)$ & first/last child of node $i$ & 
			\multicolumn{2}{c}{$\Order(\log n/\log\log n)$} \\
  $\nextsibling(i)/\prevsibling(i)$ & next/previous sibling of node $i$ & 
			\multicolumn{2}{c}{$\Order(\log n/\log\log n)$} \\
  $\subtreesize(i)$ & number of nodes in the subtree of node $i$ & 
			\multicolumn{2}{c}{$\Order(\log n/\log\log n)$} \\
  $\LA(i,d)$ & ancestor $j$ of $i$ s.t. $\depth(j) = \depth(i)-d$ &
			\multicolumn{2}{c}{$\Order(\log n)$} \\ 
  $\levelnext(i)/\levelprev(i)$ & next/previous node of $i$ in BFS order & 
			\multicolumn{2}{c}{$\Order(\log n)$} \\
  $\levelleftmost(d)/\levelrightmost(d)$ & leftmost/rightmost node with depth
$d$ & 
			\multicolumn{2}{c}{$\Order(\log n)$} \\
  $\lca(i,j)$ & the lowest common ancestor of two nodes $i,j$ & 
			\multicolumn{2}{c}{$\Order(\log n/\log\log n)$} \\
  $\deepestnode(i)$ & the (first) deepest node in the subtree of $i$ & 
			\multicolumn{2}{c}{$\Order(\log n/\log\log n)$} \\
  $\height(i)$ & the height of $i$ (distance to its deepest node) & 
			\multicolumn{2}{c}{$\Order(\log n/\log\log n)$} \\
  $\degree(i)$ & $q=$ number of children of node $i$ & 
			$\Order(q\log n/\log\log n)$ & $\Order(\log n)$ \\  
  $\child(i,q)$ & $q$-th child of node $i$ & 
			$\Order(q\log n/\log\log n)$ & $\Order(\log n)$ \\  
  $\childrank(i)$ & $q=$ number of siblings to the left of node $i$ & 
			$\Order(q\log n/\log\log n)$ & $\Order(\log n)$ \\  
  $\inorderrank(i)$ & inorder of node $i$ & 
			\multicolumn{2}{c}{$\Order(\log n/\log\log n)$}  \\
  $\inorderselect(i)$ & node with inorder $i$ & 
			\multicolumn{2}{c}{$\Order(\log n/\log\log n)$} \\
  $\leafrank(i)$ & number of leaves to the left of leaf $i$ & 
			\multicolumn{2}{c}{$\Order(\log n/\log\log n)$}  \\
  $\leafselect(i)$ & $i$-th leaf & 
			\multicolumn{2}{c}{$\Order(\log n/\log\log n)$}  \\
  $\leftmostleaf(i)/\rightmostleaf(i)$ & leftmost/rightmost leaf of node $i$ & 
			\multicolumn{2}{c}{$\Order(\log n/\log\log n)$} \\
  $\ins(i,j)$ & insert node given by matching parent.\ at $i$ and $j$ &
			$\Order(\log n/\log\log n)$ & $\Order(\log n)$ \\
  $\del(i)$ & delete node $i$ & 
			$\Order(\log n/\log\log n)$ & $\Order(\log n)$ \\
  \hline
  \end{tabular}
\end{table*}

\begin{figure}[bt]

\vspace{0.5cm}

\includegraphics[width=9.5cm]{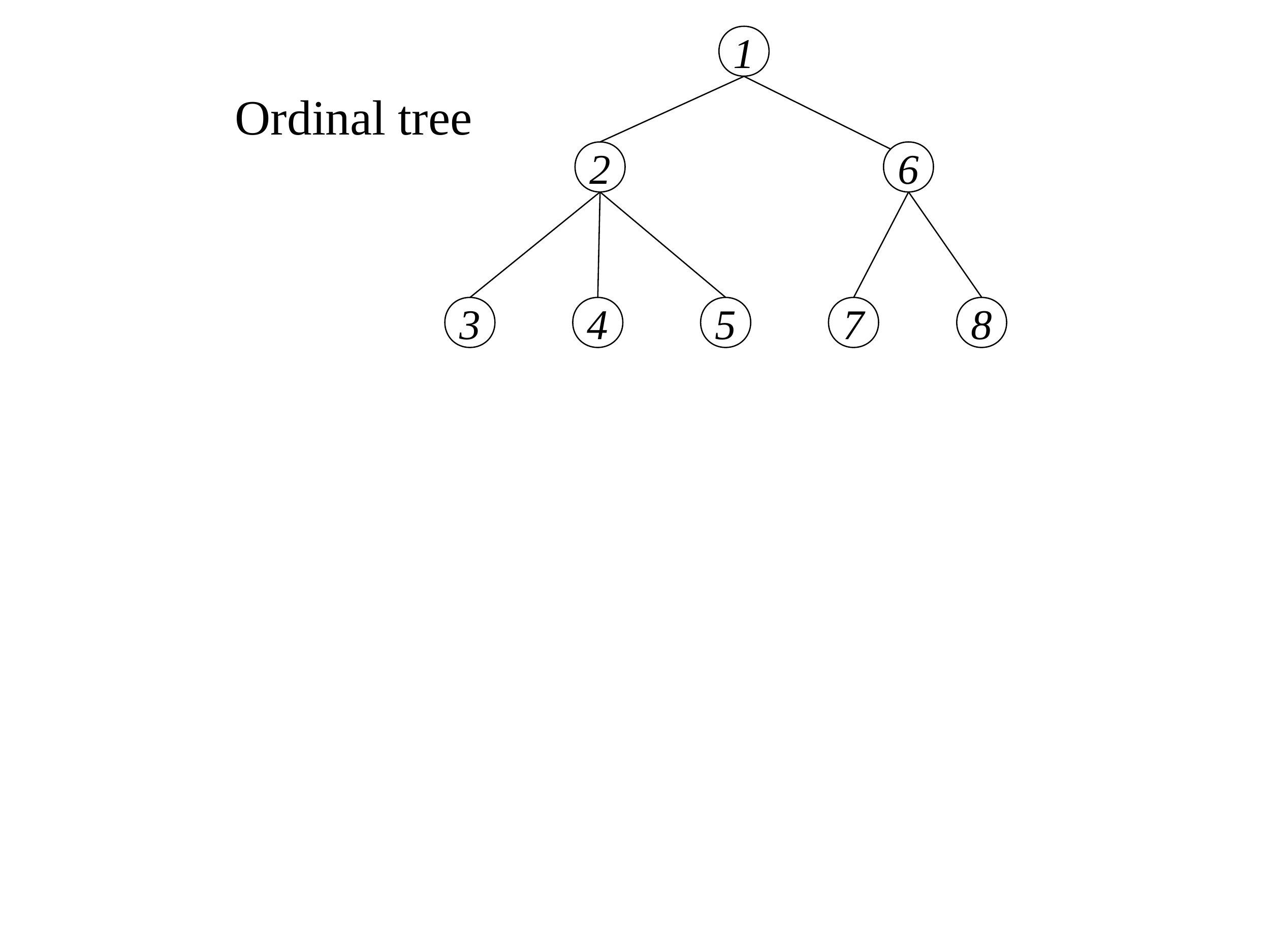} \hfill

\vspace*{-8.8cm}

\hfill \includegraphics[width=8.5cm]{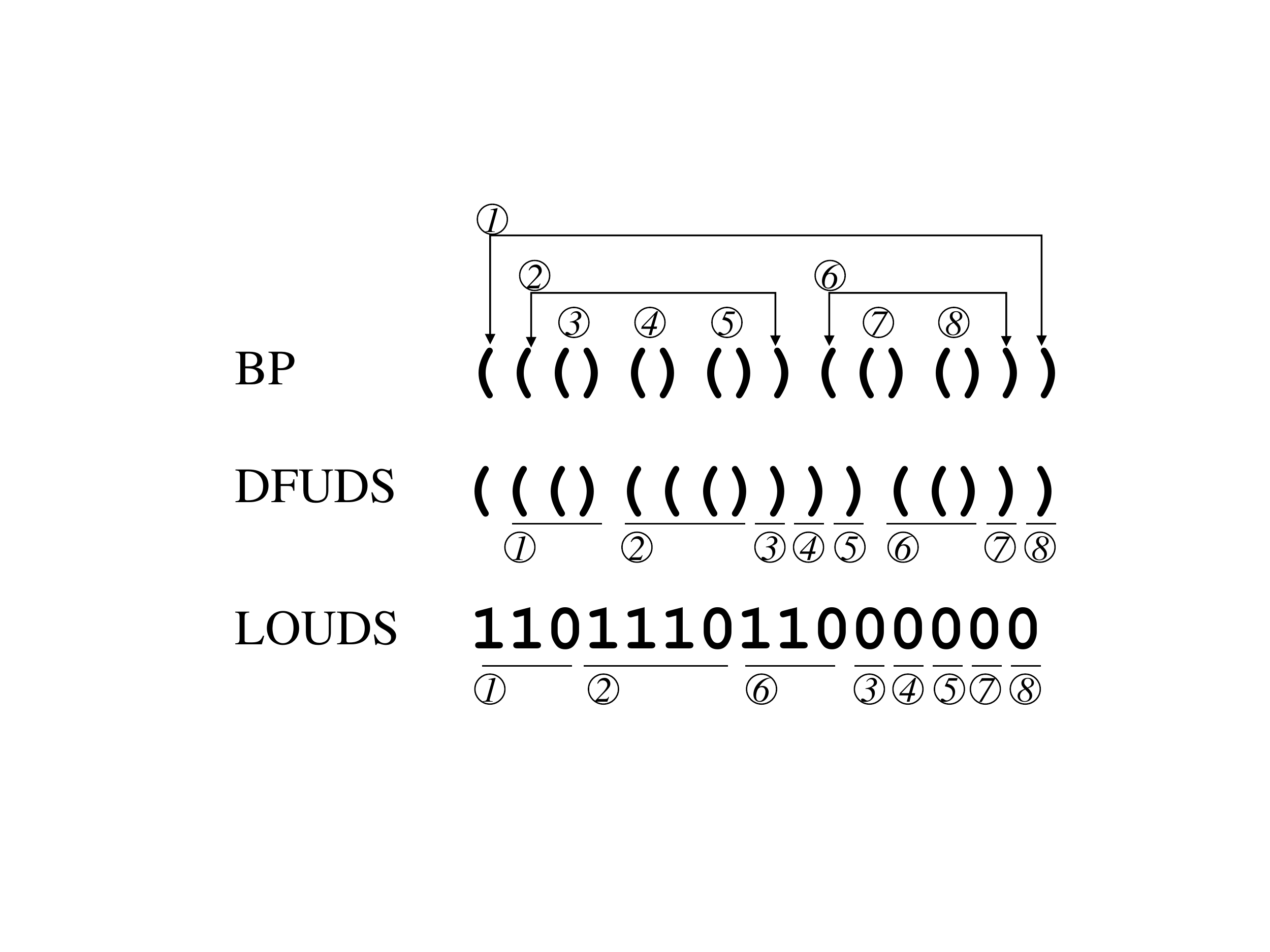}

\vspace*{-1.5cm}

\caption{Succinct representations of trees.}
\label{fig:tree}
\end{figure} 

\subsection{Organization of the paper}

In Section~\ref{sec:preliminaries} we review basic data structures
used in this paper. Section~\ref{sec:theory} describes the main ideas
for our new data structures for ordinal trees.
Sections~\ref{sec:small} and \ref{sec:large} describe the static construction.
In Sections~\ref{sec:dynamic} and \ref{sec:faster} we give two data structures
for dynamic ordinal trees.
In Section~\ref{sec:sequences} we derive our new results on compressed
sequences and applications.
In Section~\ref{sec:conclusion} we conclude and give future work directions.

\section{Preliminaries}\label{sec:preliminaries}

Here we describe the balanced parentheses sequence
and basic data structures used in this paper.

\subsection{Succinct data structures for {\rank}/{\select}}
\label{sec:rankselect}
\label{sec:mwt}

Consider a bit string $S[0,n-1]$ of length $n$.
We define {\rank} and {\select} for $S$ as follows.
$\rank_c(S,i)$ is the number of occurrences $c\in\{0,1\}$ in $S[0,i]$, 
and $\select_c(S,i)$
is the position of the $i$-th occurrence of $c$ in $S$.  Note that
$\rank_c(S,\select_c(S,i)) = i$ and $\select_c(S,\rank_c(S,i)) \le i$.
%We may omit $S$ if it is clear from the context.

There exist many succinct data structures for {\rank}/{\select}%
~\cite{Jacobson89,Munro96,RRR02}.  
A basic one uses $n+\order(n)$ bits and supports {\rank}/{\select} in 
constant time on the word RAM with word length $\Order(\log n)$.  
The space can be reduced if the number of 1's is small.  
For a string with $m$ 1's, there exists a data structure for 
constant-time {\rank}/{\select} using $nH_0(S) + \Order(n \log \log
n/\log n)$, where 
$H_0(S) = \frac{m}{n}\log\frac{n}{m}+\frac{n-m}{n}\log\frac{n}{n-m} =
m\log\frac{n}{m} + \Order(m)$
is called the empirical zero-order entropy of the sequence. The space 
overhead on top of the entropy has been recently reduced \cite{Pat08} to
$\Order(n\,t^t/\log^t n + n^{3/4})$ bits, while supporting
{\rank} and {\select} in $\Order(t)$ time. This can be built in linear
worst-case time\footnote{They use a predecessor structure by P\v{a}tra\c{s}cu 
and Thorup \cite{PT06}, more precisely their result achieving time 
``$\lg\frac{\ell-\lg\,n}{a}$'', which is a simple modification
of van Emde Boas' data structure.}.

A crucial technique for succinct data structures is {\it table lookup}.
For small-size problems we construct a table which stores answers
for all possible sequences and queries. For example, for {\rank} and {\select},
we use a table storing all answers for all 0,1 patterns of length
$\frac{1}{2} \log n$.  Because there exist only $2^{\frac{1}{2} \log n} = 
\sqrt{n}$ different patterns, we can store all answers in a universal table
(i.e., not depending on the bit sequence) that uses 
$\sqrt{n}\cdot \polylog(n) = \order(n /\polylog(n))$ bits,
which can be accessed in constant time on a word RAM
with word length $\Theta(\log n)$.

The definition of \rank\ and \select\ on bitmaps generalizes to arbitrary
sequences over an integer alphabet $[1,\sigma]$, as well as the
definition of zero-order empirical entropy of sequences, to 
$H_0(S) = \sum_{1\le c \le \sigma} \frac{n_c}{n}\log\frac{n}{n_c}$, where
$c$ occurs $n_c$ times in $S$. A compressed representation of general sequences
that supports \rank/\select\ is achieved through a structure called a
{\em wavelet tree} \cite{GroGupVit03a}. This is a complete binary tree that 
partitions the alphabet $[1,\sigma]$ into contiguous halves at each node. The 
node then stores a bitmap telling which branch did each letter go. The tree 
has height $\lceil\log\sigma\rceil$, and it reduces $\rank$ and $\select$
operations to analogous operations on its bitmap in a root-to-leaf or 
leaf-to-root traversal. If the bitmaps are represented within their 
zero-order entropy, the total space adds up to $nH_0(S)+o(n\log\sigma)$
and the operations are supported in $\Order(\log \sigma)$ time. This can
be improved to $\Order(\lceil\frac{\log\sigma}{\log\log n}\rceil)$, while
maintaining the same asymptotic space, by using a multiary wavelet tree of arity
$\Theta(\sqrt{\log n})$, and replacing the bitmaps by sequences over small
alphabets, which still can answer $\rank/\select$ in constant time 
\cite{FerManMakNav07}.

\subsection{Succinct tree representations}

A rooted ordered tree $T$, or ordinal tree, with $n$ nodes
is represented by a string $P[0,2n-1]$ of balanced parentheses of length $2n$.
A node is represented by a pair of matching parentheses $\OP \ldots \CP$
and all subtrees rooted at the node are encoded in order between the matching
parentheses
(see Figure~\ref{fig:tree} for an example).
A node $v \in T$ is identified with the position $i$ of the open parenthesis
$P[i]$ representing the node.  

There exist many succinct data structures for ordinal trees.
Among them, the ones with maximum functionality~\cite{FarMun08}
support all the operations in Table~\ref{tab:ops},
except {\ins} and {\del}, in constant time using
$2n+\Order(n \log\log\log n/\log\log n)$-bit space.
Our static data structure supports the same operations 
and reduces the space to $2n+\Order(n /\polylog(n))$ bits. %, which is optimal. 

\subsection{Dynamic succinct trees}
\label{sec:introdyn}

We consider insertion and deletion of internal nodes or leaves in ordinal trees.
In this setting, there exist no data structures supporting
all the operations in Table~\ref{tab:ops}.
The data structure of Raman and Rao~\cite{RamRao03}
supports, for binary trees, {\parent}, left and right {\em child},
and {\subtreesize} of the current node in the course of traversing
the tree in constant time, and updates in $\Order((\log\log n)^{1+\epsilon})$
time.  Note that this data structure assumes that all traversals
start from the root.
Chan et al.~\cite{CHLS07} gave a dynamic data structure
using $\Order(n)$ bits and supporting {\findopen}, {\findclose}, {\enclose},
and updates, in $\Order(\log n/\log \log n)$ time.  They also gave
another data structure using $\Order(n)$ bits and supporting {\findopen},
{\findclose}, {\enclose}, {\lca}, {\leafrank}, {\leafselect}, and updates,
in $\Order(\log n)$ time.

Furthermore, we consider the more sophisticated operation (which is simple on 
classical trees) of attaching a new subtree as the new child of a node, 
instead of just a leaf. The model is that this new subtree is already 
represented with our data structures. Both trees are thereafter blended and
become a unique tree. Similarly, we can detach any subtree from a given tree
so that it becomes an independent entity represented with our data structure.
This allows for extremely flexible support of algorithms handling dynamic trees,
far away from the limited operations allowed in previous work. This time we have
to consider a maximum possible value for $\log n$ (say, $w$, the width of the
system-wide pointers). 
Then we require $2n+\Order(n \log w/w + \sqrt{2^w})$ bits of space and
carry out the queries in time $\Order(w/\log w)$ or $\Order(w)$, depending
on the tree.  Insert or delete takes $\Order(w^{1+\epsilon})$
for any constant $\epsilon>0$ if we wish to allow attachment and detachment of
subtrees, which then can also be carried out in time $\Order(w^{1+\epsilon})$.

\subsection{Dynamic compressed bitmaps and sequences}

Let $B[0,n-1]$ be a bitmap. We want to support operations \rank\ and \select\ 
on $B$, as well as operations $\ins(B,i,b)$, which inserts bit $b$ between 
$B[i]$ and $B[i+1]$, and $\del(B,i)$, which deletes position $B[i]$ from $B$.
Chan et al.~\cite{CHLS07} handle all these operations in 
$\Order(\log n / \log\log n)$ time (which is optimal \cite{FreSak89})
using $\Order(n)$ bits of space (actually, by reducing the problem to a 
particular dynamic tree). M\"akinen and Navarro~\cite{MN08} achieve 
$\Order(\log n)$ time and $nH_0(B)+\Order(n\log\log n / \sqrt{\log n})$ bits 
of space. The results can be generalized to sequences. Gonz\'alez and Navarro
\cite{GN08} achieve 
$nH_0+\Order(n\log\sigma/\sqrt{\log n})$ bits of space and
$\Order(\log n (1+\frac{\log\sigma}{\log\log n}))$ time to handle all the
operations on a sequence over alphabet $[1,\sigma]$. They give several
applications to managing dynamic text collections, construction of static
compressed indexes within compressed space, and construction of the
Burrows-Wheeler transform \cite{BWT} within compressed space. We improve
all these results in this paper, achieving the optimal $\Order(\log n /
\log\log n)$ on polylog-sized alphabets and reducing the lower-order term in
the compressed space by a $\Theta(\log\log n)$ factor.

\section{Fundamental concepts}\label{sec:theory}

In this section we give the basic ideas of our ordinal tree representation.
In the next sections we build on these to define our static and dynamic
representations.

We represent a possibly non-balanced\footnote{As later we will use these
constructions to represent arbitrary segments of a balanced sequence.} 
parentheses sequence by a 0,1 vector $P[0,n-1]$ ($P[i] \in \{0,1\}$).  
Each opening/closing parenthesis is encoded by $\OP = 1$, $\CP = 0$. 

First, we remind that several operations of Table~\ref{tab:ops} either are 
trivial in a BP representation, or are easily solved using $\enclose$, 
$\findclose$, $\findopen$, $\rank$, and $\select$ \cite{MunRam01}. These are:
\begin{eqnarray*}
\inspect(i) &=& P[i] \textrm{ (or }\rank_1(P,i)-\rank_1(P,i-1)\textrm{ if
there is no access to }P[i] \\
\isleaf(i) &=& [P[i+1] = 0] \\
\isancestor(i,j) &=& i \le j \le findclose(P,i) \\
\depth(i) &=& \rank_1(P,i)-rank_0(P,i) \\
\parent(i) &=& \enclose(P,i) \\
\preorderrank(i) &=& \rank_1(P,i) \\
\preorderselect(i) &=& \select_1(P,i) \\
\postorderrank(i) &=& \rank_0(P,i) \\
\postorderselect(i) &=& \select_0(P,i) \\
\firstchild(i) &=& i+1 ~\textrm{(if}~P[i+1] = 1\textrm{, else}~i~\textrm{is a
	leaf)} \\ 
\lastchild(i) &=& \findopen(P,\findclose(P,i)-1) ~\textrm{(if}~P[i+1] = 1
	\textrm{, else}~i~\textrm{is a leaf)} \\ 
\nextsibling(i) &=& findclose(i)+1 ~\textrm{(if}~P[findclose(i)+1] =
	1\textrm{, else}~i~\textrm{is the last sibling)} \\
\prevsibling(i) &=& findopen(i-1) ~\textrm{(if}~P[i-1] =
	0\textrm{, else}~i~\textrm{is the first sibling)} \\
\subtreesize(i) &=& (findclose(i)-i+1)/2
\end{eqnarray*}

Hence the above operations will not be considered further in the paper.
Let us now focus on a small set of primitives needed to implement most of the
other operations.
For any function $g(\cdot)$ on $\{0,1\}$, we define the following.

\begin{definition}\label{def:sum}
For a 0,1 vector $P[0,n-1]$ and a function $g(\cdot)$ on $\{0,1\}$,
\begin{eqnarray*}
\sumw(P,g,i,j) &\mathdef& \sum_{k=i}^{j} g(P[k]) \\
\fwd(P,g,i,d) &\mathdef& \min_{j\ge i} \{j \mid \sumw(P,g,i,j) = d \} \\
\bwd(P,g,i,d) &\mathdef& \max_{j\le i} \{j \mid \sumw(P,g,j,i) = d \} \\
\rmq(P,g,i,j) &\mathdef& \min_{i \le k \le j} \{ \sumw(P,g,i,k) \} \\
\rmqi(P,g,i,j) &\mathdef& \argmin_{i \le k \le j} \{ \sumw(P,g,i,k) \} \\
\RMQ(P,g,i,j) &\mathdef& \max_{i \le k \le j} \{ \sumw(P,g,i,k) \} \\
\RMQI(P,g,i,j) &\mathdef& \argmax_{i \le k \le j} \{ \sumw(P,g,i,k) \}
\end{eqnarray*}
\end{definition}

The following function is particularly important.

\begin{definition}
Let $\pi$ be the function such that $\pi(1)=1, \pi(0)=-1$. 
Given $P[0,n-1]$, we define the {\em excess array} $E[0,n-1]$ of $P$
as an integer array such that $E[i] = \sumw(P,\pi,0,i)$. 
\end{definition}

Note that $E[i]$ stores the difference between the number of opening and 
closing parentheses in $P[0,i]$. When $P[i]$ is an opening parenthesis,
$E[i]=\depth(i)$ is the depth of the corresponding node, and is the depth minus 
1 for closing parentheses. We will use $E$ as a conceptual device in our 
discussions, it will not be stored. Note that, given the form of $\pi$, it 
holds that $|E[i+1]-E[i]| = 1$ for all $i$.

The above operations are sufficient to implement the basic navigation on
parentheses, as the next lemma shows.
Note that the equation for {\findclose} is well known, and the one
for {\LA} has appeared as well~\cite{MunRao04}, but we give proofs for completeness.

\begin{lemma}
Let $P$ be a BP sequence encoded by $\{0,1\}$. Then
{\findclose}, {\findopen}, {\enclose}, and {\LA} can be expressed as follows.
\begin{eqnarray*}
\findclose(i) &=& \fwd(P,\pi,i,0) \\
\findopen(i) &=& \bwd(P,\pi,i,0) \\
\enclose(i) &=& \bwd(P,\pi,i,2) \\
\LA(i,d) &=& \bwd(P,\pi,i,d+1)
\end{eqnarray*}
\end{lemma}

\begin{proof}
For {\findclose}, let $j>i$ be the position of the closing parenthesis matching
the opening parenthesis at $P[i]$. Then $j$ is the smallest index $>i$ such that
$E[j] = E[i]-1 = E[i-1]$ (because of the node depths). Since by definition
$E[k] = E[i-1] + \sumw(P,\pi,i,k)$ for any $k>i$, $j$ is the smallest index
$>i$ such that $\sumw(P,\pi,i,j) = 0$. This is, by definition,
$\fwd(P,\pi,i,0)$.

For {\findopen}, let $j<i$ be the position of the opening parenthesis matching
the closing parenthesis at $P[i]$. Then $j$ is the largest index $<i$ such that
$E[j-1] = E[i]$ (again, because of the node depths)\footnote{Note
$E[j]-1=E[i]$ could hold at incorrect places, where $P[j]$ is a closing
parenthesis.}. Since by definition $E[k-1] = E[i] - \sumw(P,\pi,k,i)$ for any 
$k<i$, $j$ is the largest index $<i$ such that $\sumw(P,\pi,j,i) = 0$.
This is $\bwd(P,\pi,i,0)$.

For {\enclose}, let $j<i$ be the position of the opening parenthesis that most
tightly encloses the opening parenthesis at $P[i]$. Then $j$ is the largest 
index $<i$ such that $E[j-1] = E[i]-2$ (note that now $P[i]$ is an
opening parenthesis). Now we reason as for {\findopen} to get
$\sumw(P,\pi,j,i) = 2$.

Finally, the proof for {\LA} is similar to that for {\enclose}. Now $j$ is
the largest index $<i$ such that $E[j-1] = E[i]-d-1$, which is equivalent to
$\sumw(P,\pi,j,i) = d+1$.
\end{proof}

We also have the following, easy or well-known, equalities:
\begin{eqnarray*}
\lca(i,j) &=& \max(i,j),~\textrm{if}~ isancestor(i,j)~\textrm{or}~
          			      isancestor(j,i) \\
          & & \parent(\rmqi(P,\pi,i,j)+1),~\textrm{otherwise \cite{Sada02a}} \\
\deepest(i) &=& \RMQI(P,\pi,i,\findclose(i)) \\
\height(i) &=& \depth(\deepest(i))-\depth(i) \\
\levelnext(i) &=& \fwd(P,\pi,\findclose(i),0) \\
\levelprev(i) &=& \findopen(\bwd(P,\pi,i,0)) \\
\levelleftmost(d) &=& \fwd(P,\pi,0,d) \\
\levelrightmost(d) &=& \findopen(\bwd(P,\pi,n-1,-d))
\end{eqnarray*}

We also show that the above functions unify the algorithms for computing
{\rank}/{\select} on 0,1 vectors and those for balanced parenthesis
sequences.  Namely, let $\phi, \psi$ be functions such that
$\phi(0)=0, \phi(1)=1, \psi(0)=1, \psi(1) = 0$.
Then the following equalities hold.

\begin{lemma}
For a 0,1 vector $P$,
\begin{eqnarray*}
\rank_1(P,i) &=& \sumw(P,\phi,0,i) \\
\select_1(P,i) &=& \fwd(P,\phi,0,i) \\
\rank_0(P,i) &=& \sumw(P,\psi,0,i) \\
\select_0(P,i) &=& \fwd(P,\psi,0,i) 
\end{eqnarray*}
\end{lemma}

% OJO I don't see the need for this lemma
%
%Now we prove an easy but important fact.
%Let $g(\cdot)$ be a function on $\{0,1\}$ taking values in $\{1,0,-1\}$.
%We call such a function \emph{$\pm 1$ function}.  Note that there exist
%only six such functions where $g(0)\not=g(1)$, which are indeed
%$\phi,-\phi,\psi,-\psi,\pi,-\pi$.
%For any 0,1 vector and any $\pm 1$ function, the following holds.
%\begin{lemma}\label{lem:bsearch}
%For a 0,1 vector $P$, a $\pm 1$ function $g(\cdot)$
%and integers $i \le j$, let $m = \rmq_i(P,g,i,j)$, $M = \RMQ_i(P,g,i,j)$,
%$m^- = \rmq_i(P,g,i,j-1)$, $M^- = \RMQ_i(P,g,i,j-1)$, and 
%$T = \sumw(P,g,i,j)$. Then,
%\begin{eqnarray*}
%\fwd(P,g,i,d) \le j & \iff & \sumw(P,g,i,m) \le d \le \sumw(P,g,i,M) \\
%%\bwd(P,g,j,d) \ge i & \iff & T-\sumw(P,g,m^-,j) \le d \le T-\sumw(P,g,M^-,j)
%\end{eqnarray*}
%\end{lemma}
%
%\begin{proof}
%Because the function $g(\cdot)$ takes values only 0, 1, or $-1$,
%all the integers between $\sumw(P,g,i,m)$ and $\sumw(P,g,i,M)$
%appear in the range $P[i,j]$.  Then $\fwd(P,g,i,d) \le j$ holds
%if $\sumw(P,g,i,m)$ $\le d \le \sumw(P,g,i,M)$, and vice versa.
%The reasoning for $\bwd$ is analogous; note that the right-to-left minimum
%is the total minus the left-to-right maximum and vice versa, and that
%$\bwd$ must always include cell $j$, hence the use of $m^-$ ($M^-$) instead of
%$m$ ($M$).
%\end{proof}
%
%This lemma shows that computing $\fwd$ and $\bwd$ can be reduced to
%a binary (or multi-ary) search using $\rmq_i$ and $\RMQ_i$.

Therefore, in principle we must focus only on the following set of primitives: 
{\fwd}, {\bwd}, {\sumw}, {\rmqi}, {\RMQI}, {\degree}, {\child}, and 
{\childrank}, for the rest of the paper.

Our data structure for queries on a 0,1 vector $P$ is basically a search tree
in which each leaf corresponds to a range of $P$,
and each node stores the last, maximum, and minimum values
of prefix sums for the concatenation of all the ranges up to the
subtree rooted at that node. 

\begin{definition} \label{def:minmax}
A {\em range min-max tree} for a vector $P[0,n-1]$ and a function
$g(\cdot)$ is defined as follows. Let $[\ell_1,r_1],
[\ell_2,r_2],\ldots,[\ell_q,r_q]$ be a partition of $[0,n-1]$
where $\ell_1 = 0, r_i+1 = \ell_{i+1}, r_q = n-1$.
Then the $i$-th leftmost leaf of the tree stores the sub-vector 
$P[\ell_i,r_i]$, as well as $e[i] = \sumw(P,g,0,r_i)$, 
$m[i] = e[i-1]+\rmq(P,g,\ell_i,r_i)$ and $M[i] = e[i-1]+\RMQ(P,g,\ell_i,r_i)$.
Each internal node $u$ stores in $e[u]$/$m[u]$/$M[u]$ the last/minimum/maximum 
of the $e/m/M$ values stored in its child nodes. Thus, the root node stores 
$e=\sumw(P,g,0,n-1)$, $m=\rmq(P,g,0,n-1)$ and $M=\RMQ(P,g,0,n-1)$. 
\end{definition}

\begin{example}
An example of range min-max tree is shown in Figure~\ref{fig:minmax}. Here
we use $g = \pi$, and thus the nodes store the minimum/maximum values of 
array $E$ in the corresponding interval.
\end{example}

\begin{figure}[bt]
\centerline{\includegraphics[width=10cm]{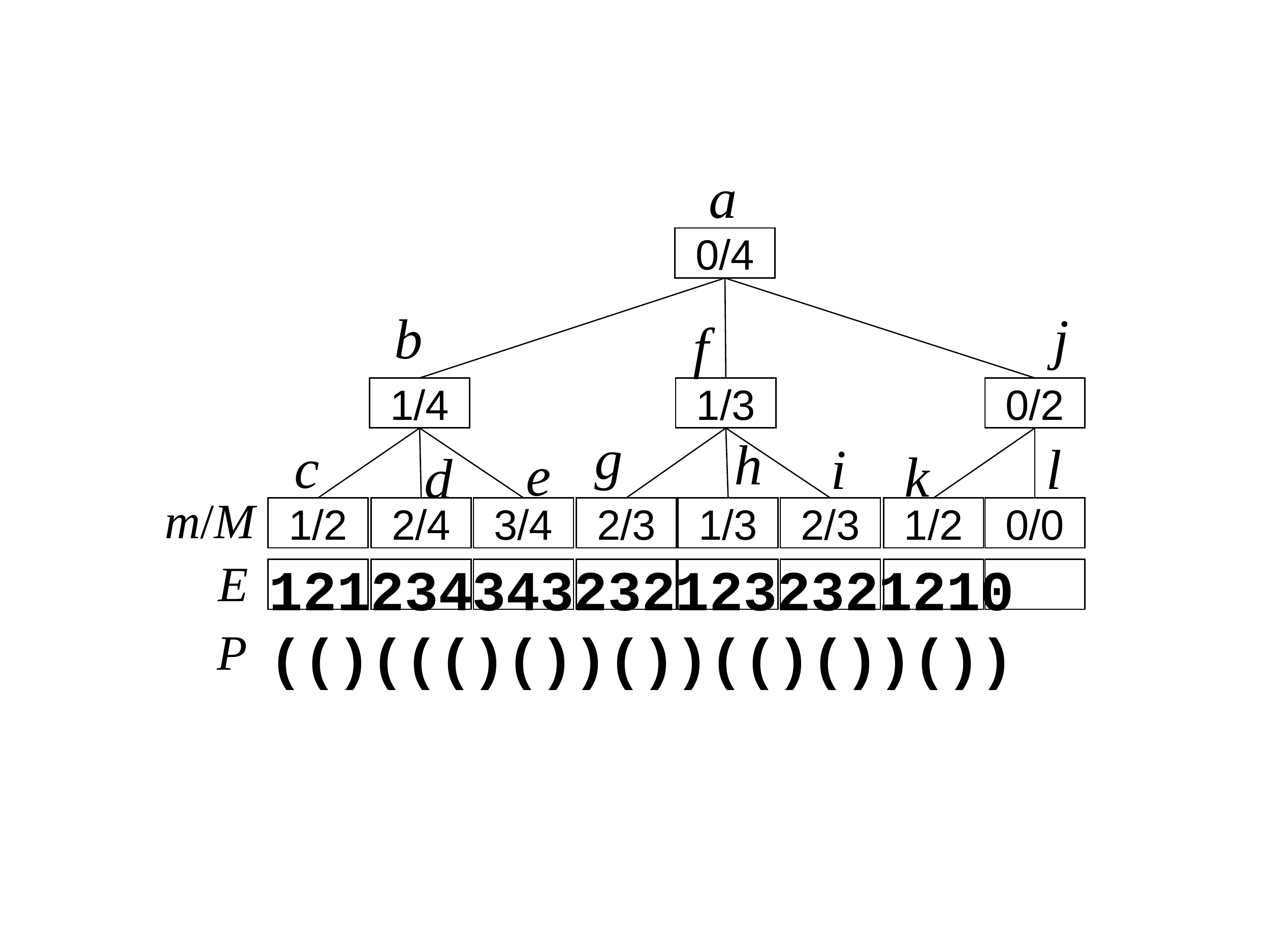}}
\vspace*{-2cm}
\caption{An example of the range min-max tree using function $\pi$,
and showing the $m/M$ values.}
\label{fig:minmax}
\end{figure} 

\section{A simple data structure for polylogarithmic-size trees}
\label{sec:small}

Building on the previous ideas, we give a simple data structure to compute 
{\fwd}, {\bwd}, and {\sumw} in constant time for arrays of polylogarithmic size.
Then we consider further operations.

Let $g(\cdot)$ be a function on $\{0,1\}$ taking values in $\{1,0,-1\}$.
We call such a function \emph{$\pm 1$ function}.  Note that there exist
only six such functions where $g(0)\not=g(1)$, which are indeed
$\phi,-\phi,\psi,-\psi,\pi,-\pi$.

Let $w$ be the bit length of the machine word in the RAM model, and $c\ge 1$ 
any constant. We have a (not necessarily balanced) parentheses vector 
$P[0,n-1]$, of moderate size $n\le N=w^c$. Assume we wish to solve the 
operations for an arbitrary $\pm 1$ function $g(\cdot)$, and let $G[i]$ 
denote $\sumw(P,g,0,i)$, analogously to $E[i]$ for $g=\pi$.

Our data structure is a range min-max tree {\TMM} for vector $P$ and function
$g(\cdot)$. Let $s = \frac{1}{2}w$. We imaginarily divide vector $P$ into
$\lceil n/s \rceil$ {\em chunks} of length $s$. These form the partition 
alluded in Definition~\ref{def:minmax}: $\ell_i = s \cdot (i-1)$. Thus the 
values $m[i]$ and $M[i]$ correspond to minima and maxima of $G$ within each 
chunk, and $e[i]=G[r_i]$.

Furthermore, the tree will be $k$-ary and complete, for $k =
\Theta(w/(c\log w))$. Thus the leaves store all the elements of arrays 
$m$ and $M$. Because it is complete, the tree can be represented just by three
integer arrays $e'[0,\Order(n/s)]$, $m'[0,\Order(n/s)]$, and 
$M'[0,\Order(n/s)]$, like a heap.

Because $-w^c \le e'[i],m'[i],M'[i] \le w^c$ for any $i$, arrays $e'$, $m'$ 
and $M'$ occupy 
$\frac{k}{k-1}\cdot \frac{n}{s} \cdot \lceil \log (2w^c+1) \rceil
= \Order(nc \log w/w)$ bits each.
The depth of the tree is $\lceil \log_k(n/s) \rceil = \Order(c)$.

The following fact is well known; we reprove it for completeness.

\begin{lemma}\label{lem:range}
Any range $[i,j] \subseteq [0,n-1]$ in {\TMM} is covered by a disjoint union
of $\Order(ck)$ subranges where the leftmost and rightmost ones
may be subranges of leaves of {\TMM}, and the others correspond to whole 
nodes of {\TMM}.
\end{lemma}

\begin{proof}
Let $a$ be the smallest value such that $i \le r_a$ and $b$ be the largest
such that $j \ge \ell_b$. Then the range $[i,j]$ is covered by the
disjoint union
$[i,j] = [i,r_a][\ell_{a+1},r_{a+1}]\ldots[\ell_b,j]$ (we can discard the
special case $a=b$, as in this case we have already one leaf covering
$[i,j]$). Then $[i,r_a]$ and $[\ell_b,j]$ are the leftmost and rightmost
leaf subranges alluded in the lemma; all the others are whole tree nodes.

It remains to show that we can reexpress this disjoint union
using $\Order(ck)$ tree
nodes. If all the $k$ children of a node are in the range, we replace the $k$
children by the parent node, and continue recursively level by level. Note that
if two parent nodes are created in a given level, then all the other 
intermediate nodes of the same level must be created as well, because the 
original/created nodes form a range at any level. At the end, there cannot
be more than $2k-2$ nodes at any level, because otherwise $k$ of them would
share a single parent and would have been replaced. As there are $c$ levels,
the obtained set of nodes covering $[i,j]$ is of size $\Order(ck)$.
\end{proof}

\begin{example}
In Figure~\ref{fig:minmax} (where $s=k=3$), the range $[3,18]$ is covered 
by $[3,5], [6,8], [9,17], [18,18]$.  They correspond to
nodes $d$, $e$, $f$, and a part of leaf $k$, respectively.
\end{example}

Computing $\fwd(P,g,i,d)$ is done as follows ({\bwd} is symmetric). First we 
check if the chunk of $i$, $[\ell_k,r_k]$ for $k=\lfloor i/s\rfloor$, contains 
$\fwd(P,g,i,d)$ with a table lookup using vector $P$, by precomputing a simple 
universal table of $2^s \log s = \Order(\sqrt{2^w}\log w)$ bits%
\footnote{Using integer division and remainder a segment within a 
chunk can be isolated and padded in constant time.
Otherwise the table is slightly larger, $2^s s^2 \log s = 
\Order(\sqrt{2^w}w^2\log w)$ bits, which will not change our final results.}
%the i is removed by shifting and padding at the right. the d, by splitting
%the chunk into 2 queries and adding 1^d (d>0) or 0^{-d} (d<0) in the
%beginning.
If so, 
we are done. Else, we compute the global target value we seek, $d'=G[i-1]+d = 
e[k]-\sumw(P,g,i,r_k)+d$ (again, the sum inside the chunk is done in
constant time using table lookup). Now we divide the range $[r_k+1,n-1]$ into 
subranges $I_1, I_2, \ldots$ represented by range min-max tree nodes $u_1, u_2,
\ldots$ as in Lemma~\ref{lem:range} (note these are simply all the right
siblings of my parent, all the right siblings of my grandparent, and so on). 
Then, for each $I_j$, we check if the target 
value $d'$ is between $m[u_j]$ and $M[u_j]$, the minimum and maximum values
of subrange $I_j$. Let $I_k$ be the first $j$ such that $m[u_j] \le d'
\le M[u_j]$, then $\fwd(P,g,i,d)$ lies within $I_k$.
If $I_k$ corresponds to an internal tree node,
we iteratively find the leftmost child of the node whose range
contains $d'$, until we reach a leaf. Finally, we find the target
in the chunk corresponding to the leaf by table lookups, using $P$ again.

\begin{example}
In Figure~\ref{fig:minmax}, where $G=E$ and $g=\pi$, computing 
$\findclose(3) = \fwd(P,\pi,3,0) = 12$
can be done as follows. Note this is equivalent to finding the first $j>3$
such that $E[i] = E[3-1]+0 = 1$. First examine the node $\lfloor 3/s\rfloor=1$ 
(labeled $d$ in the figure). We see that the target $1$ does not exist 
within $d$ after position 3. Next we examine node $e$. Since $m[e]=3$ and
$M[e]=4$, $e$ does not contain the answer either. Next we examine the node $f$.
Because $m[f]=1$ and $M[f]=3$, the answer must exist in its subtree.
Therefore we scan the children of $f$ from left to right, and find the 
leftmost one with $m[\cdot] \le 1$, which is node $h$.
Because node $h$ is already a leaf, we scan the segment corresponding
to it, and find the answer 12.
\end{example}

The sequence of subranges arising in this search corresponds to a leaf-to-leaf 
path in the range min-max tree, and it contains $\Order(ck)$ ranges according
to Lemma~\ref{lem:range}. We show now how to carry out this search in time 
$\Order(c)$ rather than $\Order(ck)$.

According to Lemma~\ref{lem:range}, the $\Order(ck)$ nodes can be partitioned
into $\Order(c)$ sequences of sibling nodes. We will manage to carry out the
search within each such sequence in $\Order(1)$ time. Assume we have to
find the first $j \ge i$ such that $m[u_j] \le d' \le M[u_j]$, where $u_1, u_2,
\ldots, u_k$ are sibling nodes in $T_{mM}$. We first check if $m[u_i] \le d'
\le M[u_i]$. If so, the answer is $u_i$. Otherwise, if $d' < m[u_i]$, the
answer is the first $j>i$ such that $m[u_j] \le d'$, and if $d' > M[u_i]$, the
answer is the first $j>i$ such that $M[u_j] \ge d'$. 

\begin{lemma} \label{lem:f}
Let $u_1,u_2,\ldots$ a sequence of $T_{mM}$ nodes containing consecutive
intervals of $P$. If $g(\cdot)$ is a $\pm 1$ function and $d < m[u_1]$,
then the first $j$ such that $d\in[m[u_j],M[u_j]]$ is the first
$j>1$ such that $d\ge m[u_j]$. Similarly, if $d>M[u_1]$, then it is the
first $j>1$ such that $d \le M[u_j]$.
\end{lemma}
\begin{proof}
Since $g(\cdot)$ is a $\pm 1$ function and the intervals are consecutive, 
$M[u_j] \ge m[u_{j-1}]-1$ and $m[u_j] \le M[u_{j-1}]+1$. Therefore, if
$d \ge m[u_j]$ and $d < m[u_{j-1}]$, then $d < M[u_j]+1$, thus
$d\in[m[u_j],M[u_j]]$; and of course $d\not\in[m[u_k],M[u_k]]$ for any $k<j$
as $j$ is the first index such that $d\ge m[u_j]$. The other case is symmetric.
\end{proof}

Thus the problem is reduced to finding the first $j>i$ such that $m[j]\le d'$,
among (at most) $k$ sibling nodes (the case $M[j]\ge d'$ is symmetric). We 
build a universal table with all the possible sequences of $k$ values $m[\cdot]$
and all possible $-w^c \le d' \le w^c$ values, and for each such 
sequence and $d'$ we store the first $j$ in the sequence such that 
$m[j] \le d'$ (or we store a mark telling that there is no such position in the 
sequence). Thus the table has $(2w^c+1)^{k+1}$ entries, and $\log(k+1)$ 
bits per entry. By choosing the constant of $k=\Theta(w/(c\log w))$
so that $k \le \frac{w}{2\log(2w^c+1)}-1$, the total space is 
$\Order(\sqrt{2^w}\log w)$ (and the arguments for the table fit in a machine 
word). With the table, each search for the first node in a sequence of siblings
can be done in constant rather than $\Order(k)$ time, and hence the overall 
time is $\Order(c)$ rather than $\Order(ck)$. Note that we store the
$m'[\cdot]$ values in heap order, and therefore the $k$ sibling values to
input to the table are stored in contiguous memory, thus they can be accessed
in constant time. We use an analogous universal table for $M[\cdot]$.

Finally, the process to solve $\sumw(P,g,i,j)$ in $\Order(c)$ time is
simple. We descend in the tree up to the leaf $[\ell_k,r_k]$ containing $j$.
We obtain $\sumw(P,g,0,\ell_k-1)=e[k-1]$ and compute the rest,
$\sumw(P,g,\ell_k,j)$, in constant time using a universal table we have already
introduced. We repeat the process for $\sumw(P,g,0,i-1)$ and then subtract 
both results.

We have proved the following lemma.

\begin{lemma}
In the RAM model with $w$-bit word size, for any constant $c\ge 1$ and
a 0,1 vector $P$ of length $n < w^c$, and a $\pm 1$ function $g(\cdot)$,
$\fwd(P,g,i,j)$, $\bwd(P,g,i,j)$, and $\sumw(P,g,i,j)$
can be computed in $\Order(c)$ time using the range min-max tree 
and universal lookup tables that require $\Order(\sqrt{2^w} \log w)$ bits.
\end{lemma}

\subsection{Supporting range minimum queries}\label{sec:rmq}

%\begin{problem}
%Given an array $A[1,n]$, a range minimum (maximum) query $\rmq(A,s,t)$ 
%($\rmq^{\max}(A,s,t)$) ($1 \le s \le t \le n$) is to return
%the index $i$ such that $s \le i \le t$
%and $A[i]$ is the minimum (maximum) value in the array.  
%If there is a tie, return the leftmost one. %(rightmost) one.
%\end{problem}

Next we consider how to compute $\rmqi(P,g,i,j)$ and $\RMQI(P,g,i,j)$.

\begin{lemma}
In the RAM model with $w$-bit word size, for any constant $c\ge 1$ and
a 0,1 vector $P$ of length $n < w^c$, and a $\pm 1$ function $g(\cdot)$,
$\rmqi(P,g,i,j)$ and $\RMQI(P,g,i,j)$
can be computed in $\Order(c)$ time using the range min-max tree 
and universal lookup tables that require $\Order(\sqrt{2^w}\log w)$ bits.
\end{lemma}

\begin{proof}
Because the algorithm for $\RMQI$ is analogous to that for $\rmqi$, we consider
only the latter. From Lemma~\ref{lem:range}, the range $[i,j]$ is covered by a 
disjoint union of $\Order(ck)$ subranges, each corresponding to some
node of the range min-max tree. Let $\mu_1,\mu_2,\ldots$ be the minimum values 
of the subranges.  Then the minimum value in $[i,j]$ is the minimum of them. 
The minimum values in each subrange are stored in array $m'$, except for at most
two subranges corresponding to leaves of the range min-max tree.  The minimum 
values of such leaf subranges are found by table lookups using $P$, by 
precomputing a universal table of $\Order(\sqrt{2^w}\log w)$ 
bits. The minimum value of a subsequence $\mu_\ell,\ldots,\mu_r$ which shares
the same parent in the range min-max tree can be also found by table lookups.
The size of such universal table is $\Order((2w^c+1)^k k\log k)
= \Order(\sqrt{2^w})$ bits (the $k$ factor is to account for queries that
span less than $k$ values, so we can specify the query length). 
Hence we find the node containing the minimum value $\mu$ among
$\mu_1,\mu_2,\ldots$, in $\Order(c)$ time. If there is a tie, we choose the 
leftmost one. 

If $\mu$ corresponds to an internal node of the range min-max tree,
we traverse the tree from the node to a leaf having the leftmost minimum
value. At each step, we find the leftmost child of the current node having
the minimum, in constant time using our precomputed table. We repeat the
process from the resulting child, until reaching a leaf. Finally, we find the 
index of the minimum value in the leaf, in constant time by a lookup on our
universal table for leaves. The overall time complexity is $\Order(c)$.
\end{proof}

\subsection{Other operations}\label{sec:degree}

The previous development on $\fwd$, $\bwd$, $\rmqi$, and $\RMQI$, has been
general, for any $g(\cdot)$. Applied to $g = \pi$, they solve a large number
of operations, as shown in Section~\ref{sec:theory}. For the remaining ones we
focus directly on the case $g = \pi$.

It is obvious how to compute $\degree(i)$, $\child(i,q)$ and $\childrank(i)$
in time proportional to the degree of the node.  To compute them in
constant time, we add another array $n'[\cdot]$ to the data structure.
In the range min-max tree, each node stores the minimum value
of a subrange for the node.  In addition to this, we store in $n'[\cdot]$
the number of the minimum values of each subrange in the tree.

\begin{lemma}
The number of children of node $i$ is equal to the number of occurrences
of the minimum value in $E[i+1,\findclose(i)-1]$.
\end{lemma}

\begin{proof}
Let $d = E[i] = \depth(i)$ and $j=\findclose(i)$. 
Then $E[j] = d-1$ and all excess values
in $E[i+1,j-1]$ are $\ge d$.  Therefore the minimum value in $E[i+1,j-1]$ is 
$d$. Moreover, for the range $[i_k,j_k]$ corresponding to the $k$-th child of 
$i$,  $E[i_k] = d+1$, $E[j_k] = d$, and all the values between them are $>d$.
Therefore the number of occurrences of $d$, which is the minimum value in 
$E[i+1,j-1]$, is equal to the number of children of $i$.
\end{proof}

Now we can compute $\degree(i)$ in constant time. Let $d = depth(i)$ and
$j = \findclose(i)$. We partition the range $E[i+1,j-1]$ into $\Order(ck)$ 
subranges, each of which corresponds to a node of the range min-max tree.
Then for each subrange whose minimum value is $d$, we sum up the number of 
occurrences of the minimum value ($n'[\cdot]$).  
The number of occurrences of the minimum value in leaf subranges
can be computed by table lookup on $P$, with a universal table using
$\Order(\sqrt{2^w}\log w)$ bits.  The time complexity is $\Order(c)$ 
if we use universal tables that let us process sequences of (up to) $k$ children
at once, that is, telling the minimum $m[\cdot]$ value within the sequence and 
the number of times it appears. This table requires 
$\Order((2w^c+1)^k k\log k) = \Order(\sqrt{2^w})$ bits.

Operation $\childrank(i)$ can be computed similarly, by counting the number
of minima in $E[parent(i),i-1]$. Operation $\child(i,q)$ follows the same idea
of $\degree(i)$, except that, in the node where the sum of $n'[\cdot]$ exceeds 
$q$, we must descend until the range min-max leaf that contains the opening
parenthesis of the $q$-th child. This search is also guided by the $n'[\cdot]$
values of each node, and is done also in $\Order(c)$ time. Here we need
another universal table that tells at which position the number of occurrences
of the minimum value exceeds some threshold, which requires
$\Order((2w^c+1)^k(2w^c+1)\log k) = \Order(\sqrt{2^w}\log w)$ bits.

For operations $\leafrank$, $\leafselect$,
$\leftmostleaf$ and $\rightmostleaf$, we define a bit-vector $P_1[0,n-1]$
such that $P_1[i] = 1 \iff P[i]=1 \land P[i+1]=0$.
Then $\leafrank(i) = \rank_1(P_1,i)$ and $\leafselect(i) = \select_1(P_1,i)$
hold. The other operations are computed by
$\leftmostleaf(i) = \select_1(P_1,\rank_1(P_1,i-1)+1)$ and
$\rightmostleaf(i) = \select_1(P_1,\rank_1(P_1,\findclose(i)))$.

We recall the definition of \emph{inorder} of nodes, which is essential for 
compressed suffix trees.

\begin{definition}[\cite{Sada07a}]\label{def:inorder}
The inorder rank of an internal node $v$ is defined as the number of visited
internal nodes, including $v$, in a left-to-right depth-first traversal,
when $v$ is visited from a child of it and another child of it
will be visited next.
\end{definition}

Note that an internal node with $q$ children has $q-1$ inorders,
so leaves and unary nodes have no inorder. We define $\inorderrank(i)$ as
the smallest inorder value of internal node $i$. 

To compute $\inorderrank$ and $\inorderselect$,
we use another bit-vector $P_2[0,n-1]$ such that
$P_2[i] = 1 \iff P[i]=0 \land P[i+1]=1$.
The following lemma gives an 
algorithm to compute the inorder of an internal node.

\begin{lemma}[\cite{Sada07a}]\label{lemma:inorder2}
Let $i$ be an internal node, and let $j=\inorderrank(i)$, so
$i=\inorderselect(j)$. Then
\begin{eqnarray*}
\inorderrank(i) &=& \rank_1(P_2,\findclose(P,i+1)) \\
\inorderselect(j) &=& \enclose(P,\select_1(P_2,j)+1)
\end{eqnarray*}

Note that $\inorderselect(j)$ will return the same node $i$ for any its
$\degree(i)-1$ inorder values.
\end{lemma}

Note that we need not to store $P_1$ and $P_2$ explicitly; they can be
computed from $P$ when needed. We only need the extra data structures for
constant-time {\rank} and {\select}, which can be reduced to the corresponding
{\sumw} and {\fwd} operations on the virtual $P_1$ and $P_2$ vectors.

\subsection{Reducing extra space}
\label{sec:small2}

Apart from vector $P[0,n-1]$, we need to store vectors $e'$, $m'$, $M'$, and 
$n'$. In addition, to implement {\rank} and {\select} using {\sumw} and {\fwd},
we would need to store vectors $e'_\phi$, $e'_\psi$, $m'_\phi$, $m'_\psi$, 
$M'_\phi$, and $M'_\psi$ which maintain the corresponding values for functions
$\phi$ and $\psi$. However, note that $\sumw(P,\phi,0,i)$ and 
$\sumw(P,\psi,0,i)$ are nondecreasing, thus the minimum/maximum within the 
chunk is just the value of the sum at the beginning/end of the chunk. Moreover,
as $\sumw(P,\pi,0,i)=\sumw(P,\phi,0,i)-\sumw(P,\psi,0,i)$ and
$\sumw(P,\phi,0,i)+\sumw(P,\psi,0,i)=i$, it turns out that both
$e_\phi[i] = (r_i+e[i])/2$ and $e_\psi[i] = (r_i-e[i])/2$ are redundant.
Analogous formulas hold for internal nodes.
Moreover, any sequence of $k$ consecutive such values can be obtained,
via table lookup, from the sequence of $k$ consecutive values of $e[\cdot]$,
because the $r_i$ values increase regularly at any node. Hence we do not store
any extra information to support $\phi$ and $\psi$.

If we store vectors $e'$, $m'$, $M'$, and $n'$ naively, we require
$\Order(nc \log w /w)$ bits of extra space on top of the $n$ bits for $P$. 

The space can be largely reduced by using a recent technique by 
P\v{a}tra\c{s}cu \cite{Pat08}. They define an {\em aB-tree} over
an array $A[0,n-1]$, for $n$ a power of $B$, as a complete tree of 
arity $B$, storing $B$ consecutive elements of $A$ in each leaf. Additionally, 
a value $\varphi \in \Phi$ is stored at each node. This must be a function of 
the corresponding elements of $A$ for the leaves, and a function of the 
$\varphi$ values of the children and of the subtree size, for internal nodes. 
The construction is able to decode the $B$ values of $\varphi$ for the children
of any node in constant time, and to decode the $B$ values of $A$ for the leaves
in constant time, if they can be packed in a machine word.

In our case, $A=P$ is the vector, $B=k=s$ is our arity, and our trees will be
of size $N=B^c$, which is slightly smaller than the $w^c$ we have been assuming.
Our values are tuples $\varphi \in 
\langle -B^c,-B^c,0,-B^c \rangle \ldots \langle B^c,B^c,B^c,B^c \rangle$ 
encoding the $m$, $M$, $n$, and $e$ values at the nodes, respectively. 
We give next their result, adapted to our case.

\begin{lemma}[adapted from Thm.\ 8 in \cite{Pat08}]\label{lem:succincter}
Let $|\Phi| = (2B+1)^{4c}$, and $B$ be such that 
$(B+1)\log(2B+1) \le \frac{w}{8c}$ (thus $B=\Theta(\frac{w}{c\log w})$).
An {\it aB-tree} of size $N = B^c$ with values in $\Phi$ can be
stored using $N+2$ bits, plus universal lookup tables of 
$\Order(\sqrt{2^w})$ bits. It can obtain the $m$, $M$, $n$ or $e$
values of the children of any node, and descend to any of those children, in 
constant time. The structure can be built in $\Order(N+w^{3/2})$ time, plus 
$\Order(\sqrt{2^w}\poly(w))$ for the universal tables.
\end{lemma}

The ``$+w^{3/2}$'' construction time comes from a fusion tree \cite{FW93} that 
is used internally on $\Order(w)$ values. It could be reduced to 
$w^\epsilon$ time for any constant $\epsilon>0$ and navigation time 
$\Order(1/\epsilon)$, but we prefer to set $c>3/2$ so that $N=B^c$ dominates
it. 

These parameters still allow us to represent our range min-max trees while
yielding the complexities we had found, as
$k=\Theta(w/(c\log w))$ 
and $N \le w^c$. Our accesses to the range min-max tree 
are either $(i)$ partitioning intervals $[i,j]$ into $\Order(ck)$ subranges,
which are easily identified by navigating from the root in $\Order(c)$ time 
(as the $k$ children are obtained together in constant time); or $(ii)$ 
navigating from the root while looking for some leaf based on the intermediate 
$m$, $M$, $n$, or $e$ values. 
Thus we retain all of our time complexities. 

The space, instead, is reduced to $N + 2 + \Order(\sqrt{2^w})$, where 
the latter part comes from our universal tables and those of
Lemma~\ref{lem:succincter} (our universal tables become smaller with the
reduction from $w$ and $s$ to $B$). 
Note that our vector $P$ must be exactly of length $N$;
padding is necessary otherwise. Both the padding and the universal tables will
lose relevance for larger trees, as seen in the next section.

The next theorem summarizes our results in this section. 

\begin{theorem} \label{th:small}
On a $w$-bit word RAM, for any constant $c>3/2$, we can represent a sequence 
$P$ of $N = B^c$ parentheses, for sufficiently small 
$B = \Theta(\frac{w}{c\log w})$,
computing all operations of Table~\ref{tab:ops} in $\Order(c)$ 
time, with a data structure depending on $P$ that uses $N + 2$ bits,
and universal tables (i.e., not depending on $P$) that use 
$\Order(\sqrt{2^w})$ bits. The preprocessing time is $\Order(N +
\sqrt{2^w}\poly(w))$ (the latter being needed only once for universal tables)
and its working space is $\Order(N)$ bits.
\end{theorem}

In case we need to solve the operations that build on $P_1$ and $P_2$, we need
to represent their corresponding $\phi$ functions (as $\psi$ is redundant).
This can still be done with Lemma~\ref{lem:succincter} using $\Phi =
(2B+1)^{6c}$ and $(B+1)\log(2B+1) \le \frac{w}{12c}$. Theorem~\ref{th:small}
applies verbatim.

\section{A data structure for large trees}\label{sec:large}

In practice, one can use the solution of the previous section for trees of any
size, achieving $\Order(\frac{k\log n}{w}\log_k n) = 
\Order(\frac{\log n}{\log w - \log\log n}) = \Order(\log n)$ time (using
$k = w/\log n$) for all operations with an extremely simple and elegant data 
structure (especially if we choose to store arrays $m'$, etc. in simple form). 
In this section we show how to achieve constant time on trees of arbitrary
size.

For simplicity, let us assume in this section that we handle trees of 
size $w^c$ in Section~\ref{sec:small}. We comment at the end the difference
with the actual size $B^c$ handled. 

For large trees with $n > w^c$ nodes, we divide the parentheses sequence into 
{\em blocks} of length $w^c$. Each block (containing a possibly non-balanced
sequence of parentheses) is handled with the range min-max tree of 
Section~\ref{sec:small}.

Let $m_1, m_2, \ldots, m_\tau$; $M_1, M_2, \ldots, M_\tau$; and $e_1, e_2, \ldots,
e_\tau$; be the minima, maxima, and excess of the $\tau = \lceil 2n/w^c \rceil$ 
blocks, respectively. These values are stored at the root nodes of each 
$T_{mM}$ tree and can be obtained in constant time.

\subsection{Forward and backward searches on $\pi$}

We consider extending $\fwd(P,\pi,i,d)$ and $\bwd(P,\pi,i,d)$ to trees of
arbitrary size. We focus on $\fwd$, as $\bwd$ is symmetric.

We first try to solve $\fwd(P,\pi,i,d)$ within the block $j = \lfloor i/w^c
\rfloor$ of $i$. If the answer is within block $j$, we are done. Otherwise,
we must look for the first excess $d' =
e_{j-1}+\sumw(P,\pi,0,i-1-w^c\cdot(j-1))+d$ in the following blocks
(where the $\sumw$ is local to block $j$). Then the answer
lies in the first block $r>j$ such that $m_r \le d' \le M_r$. Thus, we can 
apply again Lemma~\ref{lem:f}, starting at $[m_{j+1},M_{j+1}]$: If $d' \not\in 
[m_{j+1},M_{j+1}]$, we must either find the first $r>j+1$ such that 
$m_r \le j$, or such that $M_r \ge j$. Once we find such block, we complete 
the operation with a local $\fwd(P,\pi,0,d'-e_{r-1})$ query inside it.

The problem is how to achieve constant-time search, for any $j$, in a sequence
of length $\tau$. Let us focus on left-to-right minima, as the others are similar. 

\begin{definition}
Let $m_1,m_2,\ldots,m_\tau$ be a sequence of integers. We define for
each $1\le j\le \tau$ the {\em left-to-right minima starting at $j$} as
$lrm(j) = \langle j_0, j_1, j_2, \ldots \rangle$, where $j_0=j$,
$j_r < j_{r+1}$, $m_{j_{r+1}} < m_{j_r}$, and $m_{j_r+1} \ldots m_{j_{r+1}-1}
\ge m_{j_r}$.
\end{definition}

The following lemmas are immediate.

\begin{lemma} \label{lem:lrmsearch}
The first element $\le x$ after position $j$ in a sequence of 
integers $m_1,m_2,\ldots,m_\tau$ is $m_{j_r}$ for some $r>0$, where $j_r \in
lrm(j)$.
\end{lemma}

\begin{lemma} \label{lem:lrmtree}
Let $lrm(j)[p_j] = lrm(j')[p_{j'}]$. Then $lrm(j)[p_j+i] = lrm(j')[p_{j'}+i]$ 
for all $i>0$.
\end{lemma}

That is, once the $lrm$ sequences starting at two positions coincide in a
position, they coincide thereafter. Lemma ~\ref{lem:lrmtree} is essential to
store all the $\tau$ sequences $lrm(j)$ for each block $j$, in compact form.
We form a tree $T_{lrm}$, which is essentially a trie composed of the reversed 
$lrm(j)$ sequences. The tree has $\tau$ nodes, one per block. Block $j$ is a
child of block $j_1=lrm(j)[1]$ (note $lrm(j)[0]=j_0=j$), that is, $j$ is a 
child of the first block $j_1>j$ such that $m_{j_1} < m_j$. Thus each 
$j$-to-root path spells out $lrm(j)$, by Lemma~\ref{lem:lrmtree}. We add a 
fictitious root to convert the forest into a tree. Note this 
structure is called 2d-Min-Heap by Fischer \cite{Fis10}, who shows how to 
build it in linear time.

\begin{example}
Figure~\ref{fig:lrmtree} illustrates the tree built from the sequence
$\langle m_1\ldots m_9\rangle = \langle 6,4,9,7,4,4,1,8,5 \rangle$. Then
$lrm(1) = \langle 1,2,7\rangle$, $lrm(2)=\langle 2,7 \rangle$, 
$lrm(3) = \langle 3,4,5,7 \rangle$, and so on.
\end{example}

\begin{figure}[bt]
%\centerline{\epsfig{file=lrmtree.eps,width=12cm}}
\centerline{\includegraphics[width=12cm]{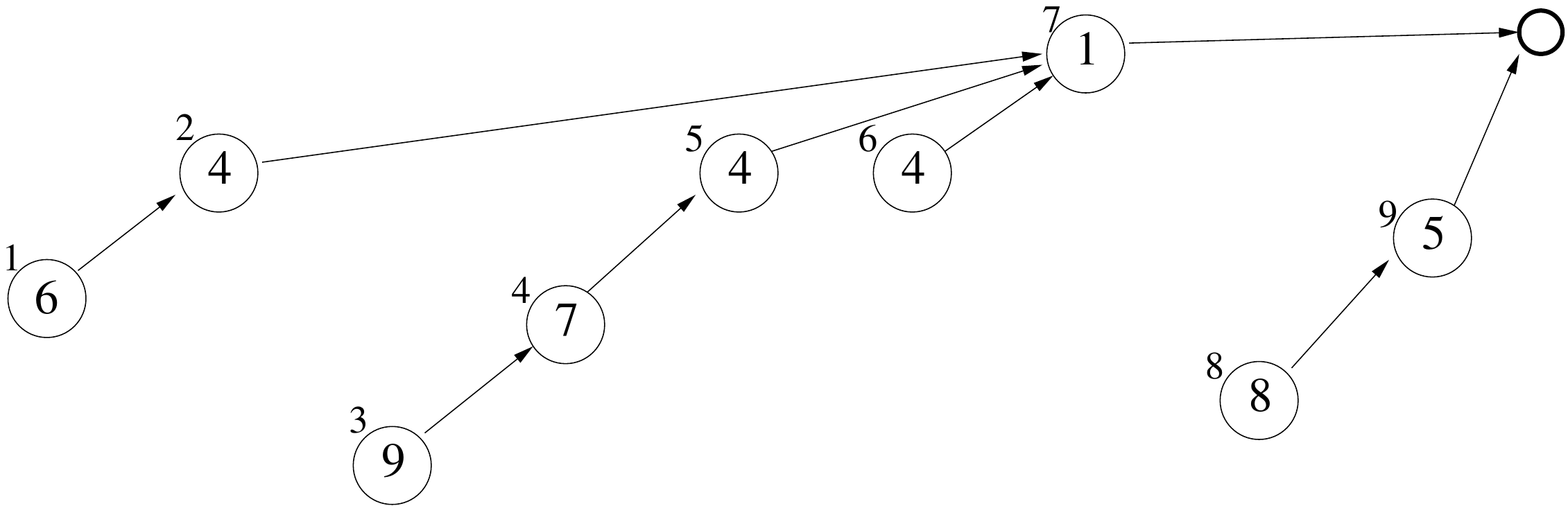}}
\caption{A tree representing the $lrm(j)$ sequences of values $m_1\ldots m_9$.}
\label{fig:lrmtree}
\end{figure} 

If we now assign weight $m_j-m_{j_1}$ to the edge between $j$ and its parent 
$j_1$, the original problem of finding the first $j_r>j$ such that 
$m_{j_r} \le d'$ reduces to finding the first ancestor $j_r$ of node $j$ such 
that the sum of the weights between $j$ and $j_r$ exceeds $d'' = m_j-d'$. 
Thus we need to compute {\em weighted level ancestors} in $T_{lrm}$.
Note that the weight of an edge in $T_{lrm}$ is at most $w^c$.

\begin{lemma}\label{lem:weighted-level-ancestor}
For a tree with $\tau$ nodes where each edge has an integer weight in $[1,W]$,
after $\Order(\tau\log^{1+\epsilon} \tau)$ time preprocessing, a weighted 
level-ancestor query is solved in $\Order(t+1/\epsilon)$ time on a 
$\Omega(\log(\tau W))$-bit word RAM. The size of the data structure is 
$\Order(\tau \log \tau \log (\tau W) + \frac{\tau W t^t}{\log^t (\tau W)} +
(\tau W)^{3/4})$ bits.
\end{lemma}

\begin{proof}
We use a variant of Bender and Farach's $\langle \Order(\tau \log \tau),\Order(1) 
\rangle$ algorithm \cite{BenFar04}. Let us ignore weights for a while. We 
extract a longest root-to-leaf path of the tree, which disconnects the tree 
into several
subtrees. Then we repeat the process recursively for each subtree, until we
have a set of paths. Each such path, say of length $\ell$, is extended
upwards, adding other $\ell$ nodes towards the root (or less if the root is
reached). The extended path is called a {\em ladder}, and its is stored as an
array so that level-ancestor queries within a ladder are trivial. This
partitioning 
guarantees that a node of height $h$ has also height $h$ in its path, and thus 
at least its first $h$ ancestors are in its ladder. Moreover the union of all 
ladders has at most $2\tau$ nodes and thus requires $\Order(\tau\log \tau)$ bits.

For each tree node $v$, an array of its (at most) $\log \tau$ ancestors at depths
$\depth(v)-2^i$, $i\ge 0$,
is stored (hence the $\Order(\tau\log \tau)$-words space and time). 
To solve the query $\LA(v,d)$, where $d'=\depth(v)-d$, the ancestor $v'$ at 
distance $d''= 2^{\lfloor \log d' \rfloor}$ from $v$ is computed. Since $v'$ 
has height at least $d''$, it has at least its first $d''$ ancestors in its
ladder. But from $v'$ we need only the ancestor at distance $d'-d'' < d''$,
so the answer is in the ladder.

To include the weights, we must be able to find the node $v'$ and the answer
considering the weights, instead of the number of nodes. We store for each 
ladder of length $\ell$ a sparse bitmap of length at most $\ell W$, where the 
$i$-th 1 left-to-right represents the $i$-th node upwards in
the ladder, and the distance between two 1s, the weight of the edge between
them. All the bitmaps are concatenated into one (so each ladder is represented
by a couple of integers indicating the extremes of its bitmap). This long
bitmap contains at most $2\tau$ 1s, and because weights do not exceed $W$, at
most $2\tau W$ 0s. Using P\v{a}tra\c{s}cu's sparse bitmaps \cite{Pat08}, it can be
represented using $\Order(\tau\log W + \frac{\tau W t^t}{\log^t (\tau W)} +
(\tau W)^{3/4})$
bits and do {\rank}/{\select} in $\Order(t)$ time.

In addition, we store for each node the $\log \tau$ accumulated weights towards 
ancestors at distances $2^i$, using fusion trees \cite{FW93}. These can store
$z$ keys of $\ell$ bits in $\Order(z\ell)$ bits and, using $\Order(z^{5/6}
(z^{1/6})^4) = \Order(z^{1.5})$ preprocessing time, answer predecessor 
queries in $\Order(\log_\ell z)$ time (via an $\ell^{1/6}$-ary tree).
The $1/6$ can be reduced to achieve $\Order(z^{1+\epsilon})$
preprocessing time and $\Order(1/\epsilon)$ query time for any desired 
constant $0<\epsilon \le 1/2$.

In our case this means $\Order(\tau\log \tau\log(\tau W))$ bits of space, 
$\Order(\tau\log^{1+\epsilon} \tau)$ construction time, and $\Order(1/\epsilon)$
access time. Thus we can find in constant time, from each node $v$,
the corresponding weighted ancestor $v'$ using a predecessor query. If this
corresponds to (unweighted) distance $2^i$, then the true ancestor is at distance
$<2^{i+1}$, and thus it is within the ladder of $v'$, where it is found using 
{\rank}/{\select} on the bitmap of ladders (each node $v$ has a pointer to its
1 in the ladder corresponding to the path it belongs to).
\end{proof}

To apply this lemma for our problem of computing $\fwd$ outside blocks,
we have $W=w^c$ and $\tau = \frac{n}{w^c}$.
Then the size of the data structure becomes 
$\Order(\frac{n\log^2 n}{w^c} + \frac{n\,t^t}{\log^t n} + n^{3/4})$.
By choosing $\epsilon = \min(1/2,1/c)$, the query time is $\Order(c+t)$ 
and the preprocessing time is $\Order(n)$ for $c > 3/2$.

\subsection{Other operations} \label{sec:otherlarge}

For computing $\rmqi$ and $\RMQI$, we use a simple data structure~%
\cite{BenFar00} on the $m_r$ and $M_r$ values, later improved to require only 
$\Order(\tau)$ bits on top of the sequence of values~\cite{Sada02a,FH07}. 
The extra space is thus $\Order(n/w^c)$ bits, and it solves any query 
up to the block granularity. For solving a general query $[i,j]$ we should compare the minimum/maximum obtained with the result of running queries 
$\rmqi$ and $\RMQI$ within the blocks at the two extremes of the boundary 
$[i,j]$. 

We consider all pairs $(i,j)$ of matching parentheses ($j = \findclose(i)$)
such that $i$ and $j$ belong to different blocks. If we define a graph
whose vertices are blocks and the edges are the pairs of parentheses 
considered, the graph is outer-planar since the parenthesis pairs nest
\cite{Jacobson89}, yet there are multiple edges among nodes. To remove these,
we choose the tightest pair of parentheses for each pair of vertices. These
parentheses are called pioneers. Since they correspond to edges of a planar 
graph, the number of pioneers is $\Order(n/w^c)$. 

For computing $\child$, $\childrank$, and $\degree$, it is enough to consider
only nodes which completely include a block (otherwise the query is solved in 
constant time by considering just two adjacent blocks; we can easily identify
such nodes using {\findclose}).
Furthermore, among them, it is enough to consider pioneers: Assume $(i,i')$
contains a whole block but is not a pioneer pair of parentheses. Then there
exists a pioneer pair $(j,j')$ contained in $(i,i')$ where $j$ is in the
same block of $i$ and $j'$ is in the same block of $i'$. Thus the block 
contains no children of $(i,i')$ as all descend from $(j,j')$. Moreover, all
the children of $(i,i')$ start either in the block of $i$ or in the block of
$i'$, since $(j,j')$ or an ancestor of it is a child of $(i,i')$. So again the 
operations are solved in constant time by considering two
blocks. Such cases can be identified by doing {\findclose} on the last child
of $i$ starting in its block and seeing if that child closes in the block of 
$i'$.

Let us call {\em marked} the nodes to consider (that is, pioneers that contain
a whole block). There are $\Order(n/w^c)$ marked nodes, thus for {\degree} we
can simply store the degrees of marked nodes using $\Order(\frac{n\log n}{w^c})$
bits of space, and the others are computed in constant time as explained.

For {\child} and {\childrank}, we set up a bitmap $C[0,2n-1]$ where marked 
nodes $v$ are indicated with $C[v]=1$, and preprocess $C$ for $\rank$ queries
so that satellite information can be associated to marked nodes. Using again 
P\v{a}tra\c{s}cu's result \cite{Pat08}, vector $C$ can be represented in at 
most $\frac{2n}{w^c}\log(w^c)+ \Order(\frac{n\,t^t}{\log^t n} + n^{3/4})$ bits, 
so that access and operation {\rank} can be computed in $\Order(t)$ time.
 
We will focus on children of marked nodes placed at the blocks fully contained 
in the nodes, as the others are in at most the two extreme blocks and can be 
dealt with in constant time. Note a block is fully contained in at most one 
marked node.

For each marked node $v$ we store a list formed by the blocks fully contained
in $v$, and the marked nodes children of $v$, in left-to-right order of $P$. 
The blocks store the number of children of $v$ that start within them, and the 
children marked nodes store simply a 1 (indicating they contain 1 child of 
$v$). All also store their position inside the list. The length of all the 
sequences adds up to $\Order(n/w^c)$ because each block and marked node appears
in at most one list. Their total sum of children is at most $n$, for the same 
reason. Thus, it is easy to store all the number of children as gaps between
consecutive 1s in a bitmap, which can be stored within the same space bounds
of the other bitmaps in this section ($\Order(n)$ bits, $\Order(n/w^c)$ 1s). 

Using this bitmap, $\child$ and $\childrank$ can easily be solved using {\rank} 
and {\select}. For $\child(v,q)$ on a marked node $v$ we start using 
$p=\rank_1(C_v,\select_0(C_v,q))$ on the bitmap $C_v$ of $v$. This tells the
position in the list of blocks and marked nodes of $v$ where the $q$-th child 
of $v$ lies. If it is a marked node, then that node is the child. If instead
it is a block $v'$, then the answer corresponds to the $q'$-th minimum within
that block, where $q'=q-\rank_0(\select_1(C_v,p))$. (Recall that we first have
to see if $\child(v,q)$ lies in the block of $v$ or in that of $\findclose(v)$,
using a within-block query in those cases, and otherwise subtracting from $q$
the children that start in the block of $v$.)

For $\childrank(u)$, we can directly store the answers for marked blocks $u$.
Else, it might be that $v=\parent(u)$ starts in the same block of $u$ or that
$\findclose(v)$ is in the same block of $\findclose(u)$, in which case we
solve $\childrank(u)$ with an in-block query and the help of $\degree(v)$.
Otherwise, the block where $u$ belongs must be in the list of $v$, say at 
position $p_u$. Then the answer is $\rank_0(C_v,\select_1(C_v,p_u))$ plus the
number of minima in the block of $u$ until $u-1$.

Finally, the remaining operations require just {\rank} and {\select} on $P$,
or the virtual bit vectors $P_1$ and $P_2$. For {\rank} it is enough to store
the answers at the end of blocks, and finish the query within a single block.
For $\select_1(P,i)$ (and similarly for $\select_0$ and for $P_1$ and $P_2$), 
we make up a sequence with the accumulated number of 1s in each of the $\tau$ 
blocks. The numbers add up to $\Order(n)$ and thus can be represented as gaps 
of 0s between consecutive 1s in a bitmap $S[0,\Order(n)]$, which can be stored 
within the previous space bounds. Computing $x = \rank_1(S,\select_0(S,i))$, 
in time $\Order(t)$, lets us know we must finish the query in block $x$, using 
its range min-max tree with the local value $i'=\select_0(S,i)-\select_1(S,x)$.

\subsection{The final result}

Recall from Theorem~\ref{th:small} that we actually use blocks of size $B^c$, 
not $w^c$, for $B = \Order(\frac{w}{c\log w})$. The sum of the space for all 
the block is $2n+ \Order(n/B^c)$, plus shared universal tables that add up to
$\Order(\sqrt{2^w})$ bits. Padding the last block to size exactly $B^c$ adds
up another negligible extra space. 

On the other hand, in this section we have extended the results to larger
trees of $n$ nodes, adding time $\Order(t)$ to the operations. By properly
adjusting $w$ to $B$ in these results, the overall extra space added is
$\Order(\frac{n(c\log B+\log^2 n)}{B^c} + \frac{n\,t^t}{\log^t n} + \sqrt{2^w}
+ n^{3/4})$ bits.
Using a computer word of $w = \log n$ bits, setting $t=c$, 
and expanding $B = \Order(\frac{\log n}{c\log\log n})$, 
we get that the time for any operation is $\Order(c)$
and the total space 
simplifies to $2n + \Order(\frac{n (c\log \log n)^c}{\log^{c-2} n})$.

Construction time is $\Order(n)$.
We now analyze the working space for constructing the data structure.
We first convert the input balanced parentheses sequence $P$ into a set of
aB-trees, each of which represents a part of the input of length $B^c$.
The working space is $\Order(B^c)$ from Theorem~\ref{th:small}.
Next we compute marked nodes:  We scan $P$ from left to right, and if
$P[i]$ is an opening parenthesis, we push $i$ in a stack, and if it is closing,
we pop an entry from the stack. At this point it is very easy to spot marked
nodes. Because $P$ is nested, the values in the stack are monotone.
Therefore we can store a new value as the difference from the previous one
using unary code. Thus the values in the stack can be stored in $\Order(n)$ 
bits. Encoding and decoding the stack values takes $\Order(n)$ time in total.
Once the marked nodes are
identified, P\v{a}tra\c{s}cu's compressed representation \cite{Pat08} of bit
vector $C$ is built in $\Order(n)$ space too, as it also cuts the bitmap into 
polylog-sized aB-trees and then computes some directories over just 
$\Order(n/\polylog(n))$ values. 

The remaining data structures, such as the {\em lrm} sequences and tree, 
the lists of the marked nodes, and the $C_v$ bitmaps, are all built on
$\Order(n/B^c)$ elements, thus they need at most $\Order(n)$ bits of space
for construction.

By rewriting $c-2-\delta$ as $c$, for any constant $\delta>0$, we get our main 
result on static ordinal trees, Theorem~\ref{th:main}. 

\section{A simple data structure for dynamic trees}\label{sec:dynamic}

In this section we give a simple data structure for dynamic ordinal trees.
In addition to the previous query operations, 
we add now insertion and deletion of internal nodes and leaves.

\subsection{Memory management}
\label{sec:memory}

We store a 0,1 vector $P[0,2n-1]$ using a dynamic min-max tree. Each leaf of 
the min-max tree stores a {\em segment} of $P$ in verbatim form. The length $\ell$ 
of each segment is restricted to $L \le \ell \le 2L$ for some parameter $L>0$.

If insertions or deletions occur, the length of a segment will change.
We use a standard technique for dynamic maintenance of memory
cells~\cite{Munro86}. We regard the memory as an array of cells of length $2L$ 
each, hence allocation is easily handled in constant time.
We use $L+1$ linked lists $s_L,\ldots,s_{2L}$ where
$s_i$ stores all the segments of length $i$. All the segments with equal length
$i$ are packed consecutively, without wasting any extra space in the cells of
linked list $s_i$ (except possibly at the head cell of each list). 
Therefore a cell (of length $2L$) stores (parts of)
at most three segments, and a segment spans at most two cells.
Tree leaves store pointers to the cell and offset where its segment is stored.
If the length of a segment changes from $i$ to $j$, it is moved
from $s_i$ to $s_j$. The space generated by the removal is 
filled with the head segment in $s_i$, and the removed segment is stored
at the head of $s_j$. 

With this scheme, scanning any segment takes $\Order(L/\log n)$ time, by 
processing it by chunks of $\Theta(\log n)$ bits. This is also the time to
compute operations $\fwd$, $\bwd$, $\rmqi$, etc. on the segment, using 
universal tables. Migrating a node to another list is also done in 
$\Order(L/\log n)$ time.

If a migration of a segment occurs, pointers to
the segment from a leaf of the tree must change. For this sake
we store back-pointers from each segment to its leaf. Each cell stores also
a pointer to the next cell of its list. Finally, an
array of pointers for the heads of $s_L,\ldots,s_{2L}$ 
is necessary. 
Overall, the space for storing a 0,1 vector
of length $2n$ is $2n + \Order(\frac{n \log n}{L})$ bits.

The rest of the dynamic tree will use sublinear space, and thus we 
allocate fixed-size memory cells for the internal nodes, as they will waste
at most a constant fraction of the allocated space. 

\subsection{A dynamic tree}\label{sec:simple}

We give a simple dynamic data structure representing an ordinal
tree with $n$ nodes using $2n+\Order(n/\log n)$ bits, and supporting
all query and update operations in $\Order(\log n)$ worst-case time.

We divide the 0,1 vector $P[0,2n-1]$ into segments of length from $L$
to $2L$, for $L = \log^2 n$. We use a balanced binary tree for
representing the range min-max tree. If a node of the tree
corresponds to a vector $P[i,j]$, the node stores $i$ and $j$, as well as
$e = \sumw(P,\pi,i,j)$, $m = \rmq(P,\pi,i,j)$, $M = \RMQ(P,\pi,i,j)$, 
and $n$, the number of minimum values in $P[i,j]$ regarding $\pi$. (Data on
$\phi$ for the virtual vectors $P_1$ and $P_2$ is handled analogously.)

It is clear that $\fwd$, $\bwd$, $\rmqi$, $\RMQI$, $\rank$, $\select$,
$\degree$, $\child$ and $\childrank$ can be computed
in $\Order(\log n)$ time, by using the same algorithms developed for small
trees in Section~\ref{sec:small}. These operations cover all the
functionality of Table~\ref{tab:ops}. Note the values we store are local to
the subtree (so that they are easy to update), but global values are easily 
derived in a top-down traversal. For example, to solve $\fwd(P,\pi,i,d)$
starting at the min-max tree root $v$ with children $v_l$ and $v_r$, we first 
see if $j(v_l) \ge i$, in which case try first on $v_l$. If the answer is not 
there or $j(v_l)<i$, we try on $v_r$, changing $d$ to $d-e(v_l)$. This will only
traverse $\Order(\log n)$ nodes, as seen in Section~\ref{sec:small}. As another
example, to compute $\depth(i)$ from $v$ we first see if $j(v_l) \ge i$, in 
which case we continue at $v_l$, otherwise we continue at $v_r$ and add
$e(v_l)$ to that result.

Because each node uses $\Order(\log n)$ bits,
and the number of nodes is $\Order(n/L)$, the total space is
$2n+\Order(n/\log n)$ bits. This includes the extra
$\Order(\frac{n\log n}{L})$ term for the leaf data.
Note that we need to maintain several universal tables that handle chunks of
$\frac{1}{2}\log n$ bits. These require $\Order(\sqrt{n}\cdot
\polylog(n))$
extra bits, which is negligible.

If insertion/deletion occurs, we update a segment, and the stored
values in the leaf for the segment. From the leaf we step back to the root,
updating the values as follows:
\begin{eqnarray*}
i(v),j(v) &=& i(v_l), j(v_r) \\
e(v) &=& e(v_l)+e(v_r) \\
m(v) &=& \min(m(v_l),e(v_l)+m(v_r)) \\
M(v) &=& \max(M(v_l),e(v_l)+M(v_r)) \\
n(v) &=& n(v_l) ~\textrm{if}~m(v_l) < e(v_l)+m(v_r), \\
     & & n(v_r) ~\textrm{if}~m(v_l) > e(v_l)+m(v_r), \\
     & & n(v_l)+n(v_r) ~\textrm{otherwise.}
\end{eqnarray*}

If the length of the segment exceeds $2L$, we split it into two and add a 
new node. If, instead, the length
becomes shorter than $L$, we find the adjacent segment to the right.
If its length is $L$, we concatenate them; otherwise move the
leftmost bit of the right segment to the left one.  In this manner
we can keep the invariant that all segments have length $L$ to $2L$.
Then we update all the values in the ancestors of the modified leaves, 
as explained.  If a balancing operation occurs,
we also update the values in nodes. All these updates are carried out
in constant time per involved node, as their values are recomputed using the
formulas above. Thus the update time is also $\Order(\log n)$.

When $\lceil \log n\rceil$ changes, we must update the allowed values for $L$, 
recompute universal tables, change the width of the stored values, etc. 
M\"akinen and
Navarro \cite{MN08} have shown how to do this for a very similar case (dynamic
{\rank}/{\select} on a bitmap). Their solution of splitting the bitmap into 
three parts and moving border bits across parts to deamortize the work applies
verbatim to our sequence of parentheses, thus we can handle changes in 
$\lceil \log n\rceil$ without altering
the space nor the time complexity (except for $\Order(w)$ extra bits in the
space due to a constant number of system-wide pointers, a technicism we ignore).
We have one range min-max tree for each of the three parts and adapt all the 
algorithms in the obvious manner\footnote{One can
act as if one had a single range min-max tree where the first two levels were 
used to split the three parts (these first nodes would be special in the sense 
that their handling of insertions/deletions would reflect the actions on moving
bits between the three parts).}.

\section{A faster dynamic data structure}\label{sec:faster}

Instead of the balanced binary tree, we use a B-tree with branching factor 
$\Theta(\sqrt{\log n})$, as in previous work \cite{CHLS07}. Then 
the depth of the tree is $\Order(\log n/\log\log n)$. The lengths of segments
is $L$ to $2L$ for $L=\log^2 n/\log\log n$. The required space for the range 
min-max tree and the vector is now $2n+\Order(n \log\log n/\log n)$ bits (the
internal nodes use $\Order(\log^{3/2} n)$ bits but there are only
$\Order(\frac{n}{L\sqrt{\log n}})$ of them).
Now each leaf can be processed in time $\Order(\log n/\log\log n)$.

Each internal node $v$ of the range min-max tree has $k$ children,
for $\sqrt{\log n} \le k \le 2\sqrt{\log n}$ (we relax the constants later).
Let $c_1,c_2,\ldots,c_k$ be the children of $v$,
and $[\ell_1,r_1],\ldots,[\ell_k,r_k]$ be their corresponding subranges.
We store 
$(i)$ the children boundaries $\ell_i$, 
$(ii)$ $s_\phi[1,k]$ and $s_\psi[1,k]$ storing $s_{\phi/\psi}[i] = 
\sumw(P,\phi/\psi,\ell_1,r_i)$, 
$(iii)$ $e[1,k]$ storing $e[i] = \sumw(P,\pi,\ell_1,r_i)$, 
$(iv)$ $m[1,k]$ storing $m[i] = e[i-1]+\rmq(P,\pi,\ell_i,r_i)$, 
$M[1,k]$ storing $M[i] = e[i-1]+\RMQ(P,\pi,\ell_i,r_i)$, and
$(v)$ $n[1,k]$ storing in $n[i]$ the number of times the minimum excess within
the $i$-th child occurs within its subtree.
Note that the values stored are local to the subtree (as in the simpler 
balanced binary tree version, Section~\ref{sec:dynamic}) 
but cumulative with respect to previous siblings.
Note also that storing $s_\phi$, $s_\psi$ and $e$ is redundant, as noted in
Section~\ref{sec:small2}, but we need $s_{\phi/\psi}$ in explicit form to 
achieve constant-time searching into their values, as it will be clear soon. 

Apart from simple accesses to the stored values, we need to support the 
following operations within any node:
\begin{itemize}
\item $p(i)$: the largest $j$ such that $\ell_{j-1} \le i$ (or $j=1$).
\item $w_{\phi/\psi}(i)$: the largest $j$ such that $s_{\phi/\psi}[j-1] \le i$ 
(or $j=1$).
\item $f(i,d)$: the smallest $j\ge i$ such that $m[j] \le d \le M[j]$.
\item $b(i,d)$: the largest $j\le i$ such that $m[j] \le d \le M[j]$.
\item $r(i,j)$: the smallest $x$ such that $m[x]$ is minimum in $m[i,j]$.
\item $R(i,j)$: the smallest $x$ such that $m[x]$ is maximum in $m[i,j]$.
\item $n(i,j)$: the number of times the minimum within the subtrees
of children $i$ to $j$ occurs within that range.
\item $r(i,j,t)$: the $x$ such that the $t$-th minimum within the subtrees
of children $i$ to $j$ occurs within the $x$-th child. 
\item {\em update}: updates the data structure upon $\pm 1$ changes in
some child.
\end{itemize}

Simple operations involving {\rank} and {\select} on $P$ are carried out
easily with $\Order(\log n / \log\log n)$ applications of $p(i)$ and 
$w_{\phi/\psi}(i)$. For example $\depth(i)$ is computed, starting from the
root node, by finding the child $j=p(i)$ to descend, then recursively computing 
$\depth(i-\ell_j)$ on the $j$-th child, and finally adding $e[j-1]$ to the 
result. Handling $\phi$ for $P_1$ and $P_2$ is immediate; we omit it.

Operations $\fwd$/$\bwd$ can be carried out via $\Order(\log n /
\log\log n)$ applications of $f(i,d)/b(i,d)$. Recalling Lemma~\ref{lem:range},
the interval of interest is partitioned into 
$\Order (\sqrt{\log n} \cdot \log n / \log\log n)$ nodes of the B-tree, but 
these can be grouped into $\Order(\log n / \log\log n)$ sequences of 
consecutive siblings. Within each such sequence a single $f(i,d)/b(i,d)$ 
operation is sufficient. For example, for $\fwd(i,d)$, let us assume $d$ is
a global excess to find (i.e., start with $d \leftarrow d + \depth(i) - 1$).
We start at the root $v$ of the range min-max tree, and compute $j = p(i)$, so
the search starts at the $j$-th child, with the recursive query
$\fwd(i-\ell_j,d-e[j-1])$. If the answer is not found in that child, query
$j' = f(j+1,d)$ tells that it is within child $j'$. We then enter recursively
into the $j'$-th child of the node with $\fwd(i-\ell_{j'},d-e[j'-1])$, where the
answer is sure to be found.

Operations $\rmqi$ and $\RMQI$ are solved in very similar fashion, using 
$\Order(\log n / \log\log n)$ applications of $r(i,j)/R(i,j)$. For
example, to compute $\rmq(i,i')$ (the extension to $\rmqi$ is obvious) we 
start with $j=p(i)$ and $j'=p(i')$. If $j=j'$ we answer with 
$e[j-1]+\rmq(i-\ell_j,i'-\ell_j)$ on the $j$-th child of 
the current node. Otherwise we recursively compute 
$e[j-1]+\rmq(i-\ell_j,\ell_{j+1}-\ell_j-1)$,
$e[j'-1]+\rmq(0,i'-\ell_{j'})$ and, if $j+1<j'$, $m[r(j+1,j'-1)]$,
and return the minimum of the two or three values.

For {\degree} we partition the interval as for $\rmqi$ and then use 
$m[r(i,j)]$ in each node to identify those holding the global minimum. For 
each node holding the minimum, $n(i,j)$ gives the number of occurrences of the
minimum in the node. Thus we apply $r(i,j)$ and $n(i,j)$ 
$\Order(\log n / \log\log n)$ times. Operation {\childrank} is very similar, 
by changing the right end of the interval of interest, as before. Finally, 
solving {\child} is also similar, except that when we exceed the desired rank 
in the sum (i.e., in some node $n(i,j) \ge t$, where $t$ is the local rank of 
the child we are looking for), we find the desired min-max tree branch with 
$r(i,j,t)$, and continue on the child with $t \leftarrow t-n(i,r(i,j,t)-1)$,
using one $r(i,j,t)$ operation per level.

\subsection{Dynamic partial sums}

Let us now face the problem of implementing the basic operations. Our first
tool is a result by Raman et al., which solves several subproblems of the same
type.

\begin{lemma}[\cite{RRR01}] \label{lem:partialsums}
Under the RAM model with word size $\Theta(\log n)$, it is possible to maintain
a sequence of $\log^\epsilon n$ nonnegative integers $x_1,x_2,\ldots$ of 
$\log n$ bits each, for any constant $0\le \epsilon<1$, such that the data 
structure requires $\Order(\log^{1+\epsilon} n)$ bits and carries out the 
following operations in constant time: $sum(i) = \sum_{j=1}^i x_j$, 
$search(s) = \max \{i,~sum(i) \le s\}$, and $update(i,\delta)$, which sets 
$x_i \leftarrow x_i+\delta$, for $-\log n \le \delta \le\log n$.
The data structure also uses a precomputed universal table of size 
$\Order(n^{\epsilon'})$ bits for any fixed $\epsilon' > 0$.
The structure can be built in $\Order(\log^\epsilon n)$ time except the table.
\end{lemma}

Then we can store $\ell$, $s_\phi$, and $s_\psi$ in differential form, and 
obtain their values via $sum$. The same can be done with $e$, provided we fix
the fact that it can contain negative values by storing 
$e[i]+2^{\lceil \log n\rceil}\cdot i$ (this
works for constant-time $sum$, yet not for $search$).
Operations $p$ and $w_{\phi/\psi}$ are then solved via $search$ on
$\ell$ and $s$, respectively. Moreover we can handle $\pm 1$ changes in
the subtrees in constant time as well. In addition, we can store $m[i]-e[i-1]$
and $M[i]-e[i-1]$, which depend only on the subtree, and reconstruct the
values in constant time using $sum$ on $e$, which eliminates the problem of
propagating changes in $e[i]$ to $m[i+1,k]$ and $M[i+1,k]$. Local changes to
$m[i]$ or $M[i]$ can be applied directly.

\subsection{Cartesian trees}

Our second tool is the Cartesian tree~\cite{Vui80,Sad06a}. A Cartesian tree
for an array $B[1,k]$ is a binary tree in which the root node 
stores the minimum value $B[\mu]$, and the left and the right subtrees are 
Cartesian trees for $B[1,\mu-1]$ and $B[\mu+1,k]$, respectively. If there 
exist more than one minimum value position, then $\mu$ is the leftmost.
Thus the tree shape has enough information to determine the 
position of the leftmost minimum in any range $[i,j]$. As it is a binary
tree of $k$ nodes, a Cartesian tree can be represented within $2k$ bits 
using parentheses and the bijection with general trees. It can be built in
$\Order(k)$ time.

We build Cartesian trees for $m[1,k]$ and for $M[1,k]$ (this one taking
maxima). Since $2k=\Order(\sqrt{\log n})$, universal tables let us answer in 
constant time any query of the form $r(i,j)$ and $R(i,j)$, as these depend 
only on the tree shape as explained. 
All the universal tables we will use on Cartesian trees take
$\Order(2^{\Order(\sqrt{\log n})}\cdot\polylog(n)) = o(n^\alpha)$ for any
constant $0<\alpha<1$.

We also use Cartesian trees to solve operations $f(i,d)$ and $b(i,d)$. 
However, these do not depend only on the tree shape, but on the actual values 
$m[i,k]$. We focus on $f(i,d)$ since $b(i,d)$ is symmetric. Following 
Lemma~\ref{lem:f}, we
first check whether $m[i] \le d \le M[i]$, in which case the answer is $i$. 
Otherwise, the answer is either the next $j$ such that $m[j] \le d$ (if
$d<m[i]$), or $M[j] \ge d$ (if $d>M[i]$). Let us focus on the case $d<m[i]$, 
as the other is symmetric. By Lemma~\ref{lem:lrmsearch}, the answer belongs 
to $lrm(i)$, where the sequence is $m[1,k]$.

%This search could be carried out in constant time on the SC-tree
%(with a universal table) if we knew the rank of $d$ in the set $m[1,k]$. That
%is, if we knew that $d$ is larger or equal than $j$ values $m[\cdot]$ and 
%strictly smaller than $k-j$, then $f(i,d)$ is the first position of any of the 
%$j$ smallest $m[1,k]$ values in the range $m[i,k]$. This depends entirely on
%the shape of the SC-tree for $m[1,k]$.
%
%To find out the rank of any $d$, we maintain a permutation $\Pi_m[1,k]$,
%which puts the $m[1,k]$ values in sorted order, and represent the reordered
%$m$ values in differential form using the structure of
%Lemma~\ref{lem:partialsums}. Then $m[i] = sum(\Pi[i])$. Now the rank of any 
%$d$ is simply $search(d)$, and the operation can be carried out in constant
%time.
%
%OJO updates don't work, many values change upon s[] inc/dec, pity...

\begin{lemma}
Let $C$ be the Cartesian tree for $m[1,k]$. Then $lrm(i)$ is the sequence of 
nodes of $C$ in the upward path from $i$ to the root, which are reached from 
the left child.
\end{lemma}
\begin{proof}
The left and right children of node $i$ contain values not smaller than $i$. 
All the nodes in the upward path are equal to or smaller than $i$. Those 
reached from the right must be at the left of position $i$, as they must be 
either to the left or to the right of all the nodes already seen, and $i$ has 
been seen. Their left children are also to the left of $i$. Ancestors $j$ 
reached from the left are strictly smaller than $i$ and, by the previous
argument, to the right of $i$, thus they belong to $lrm(i)$. Finally, the 
right descendants of those $j$ are not in $lrm(i)$ because they are 
after $j$ and equal to or larger than $m[j]$.
\end{proof}

The Cartesian tree can have precomputed $lrm(i)$ for each $i$, as this 
depends only on the tree shape, and thus are stored in universal tables. This 
is the sequence of positions in $m[1,k]$ that must be considered. We can then 
binary search this sequence, using the technique described to retrieve any 
desired $m[j]$, to compute $f(i,d)$ in $\Order(\log k)=\Order(\log\log n)$ time.

\subsection{Complete trees}

We arrange a complete binary tree on top of the $n[1,k]$ values, so that each
node of the tree records $(i)$ one leaf where the subtree minimum is
attained, and $(ii)$ the number of times the minimum arises in its subtree.
This tree is arranged in heap order and requires $\Order(\log^{3/2} n)$
bits of space.

A query $n(i,j)$ is answered essentially as in Section~\ref{sec:dynamic}: We
find the $\Order(\log k)$ nodes that cover $[i,j]$, find the minimum $m[\cdot]$
value among the leaves stored in $(i)$ for each covering node (recall we have
constant-time access to $m$), and add up the number of times (field $(ii)$)
the minimum of $m[i,j]$ occurs. This takes overall $\Order(\log k)$ time.

A query $r(i,j,t)$ is answered similarly, stopping at the node where the
left-to-right sum of the fields $(ii)$ reaches $t$, and then going down to
the leaf $x$ where $t$ is reached. Then the $t$-th occurrence of the minimum
in subtrees $i$ to $j$ occurs within the $x$-th subtree.

When an $m[i]$ or $n[i]$ value changes, we must update the upward path towards
the root of the complete tree, using the update formula for $n(v)$ given in
Section~\ref{sec:dynamic}. This is also sufficient when $e[i]$ changes: Although
this implicitly changes all the $m[i+1,k]$ values, the local subtree data 
outside the ancestors of $i$ are unaffected. Then the root $n(v)$ value will
become an $n[i']$ value at the parent of the current range min-max tree node
(just as the minimum of $m[1,k]$, maximum of $M[1,k]$, excess $e[k]$, etc.,
which can be computed in constant time as we have seen).

Since these operations take time $\Order(\log k) = \Order(\log\log n)$ time,
the time complexity of {\degree}, {\child}, and {\childrank} is 
$\Order(\log n)$. Update operations ({\ins} and {\del}) also require
$\Order(\log n)$ time, as we may need to update $n[\cdot]$ for one node per
tree level. However, as we see later, it is possible to achieve
time complexity $\Order(\log n / \log\log n)$ for {\ins} and {\del} for all
the other operations. Therefore, we might choose not to support operations
$n(i,j)$ and $r(i,j,t)$ to retain the lower update complexity. In this case,
operations {\degree}, {\child}, and {\childrank} can only be implemented
naively using {\firstchild}, {\nextsibling}, and {\parent}.

\subsection{Updating Cartesian trees}

We already solved some simple cases of update, but not yet how to maintain
Cartesian trees.
When a value $m[i]$ or $M[i]$ changes (by $\pm 1$), the Cartesian trees
might change their shape. Similarly, a $\pm 1$ change in $e[i]$ induces a change
in the effective value of $m[i+1,k]$ and $M[i+1,k]$. We store $m$ and $M$ in
a way independent of $e$, but the Cartesian trees are built upon the
actual values of $m$ and $M$. Let us focus on $m$, as $M$ is similar.
If $m[i]$ decreases by 1, we need to determine if $i$ should go higher in
the tree. We compare $i$ with its Cartesian tree parent $j = Cparent(i)$ and, 
if $(a)$ $i < j$ and $m[i]-m[j]=0$, or if $(b)$ $i>j$ and $m[i]-m[j]=-1$, we 
must carry out a rotation with $i$ and $j$. Figure~\ref{fig:mdecr} shows the
two cases. As it can be noticed, case $(b)$ may propagate the rotations
towards the new parent of $i$, as it generates a new distance $d-1$ that is
smaller than before.

In order to carry out those propagations in constant time, we store an array
$d[1,k]$, so that $d[i] = m[i]-m[Cparent(i)]$ if this is $\le k+2$, and $k+2$
otherwise. Since $d[1,k]$ requires $\Order(k\log k) = \Order(\sqrt{\log n}\,
\log\log n) = o(\log n)$ bits of space, it can be manipulated in constant time
using universal tables: With $d[1,k]$ and the current Cartesian tree as input,
a universal table can precompute the outcome of the changes in $d[\cdot]$ and
the corresponding sequence of rotations triggered by the decrease of $m[i]$
for any $i$, so we can obtain in constant time the new Cartesian tree and the 
new table $d[1,k]$. The limitation of values up to $k+2$ is necessary for the 
table fitting in a machine word, and its consequences will be discussed soon.

\begin{figure}[bt]
\centerline{\includegraphics[width=0.67\textwidth]{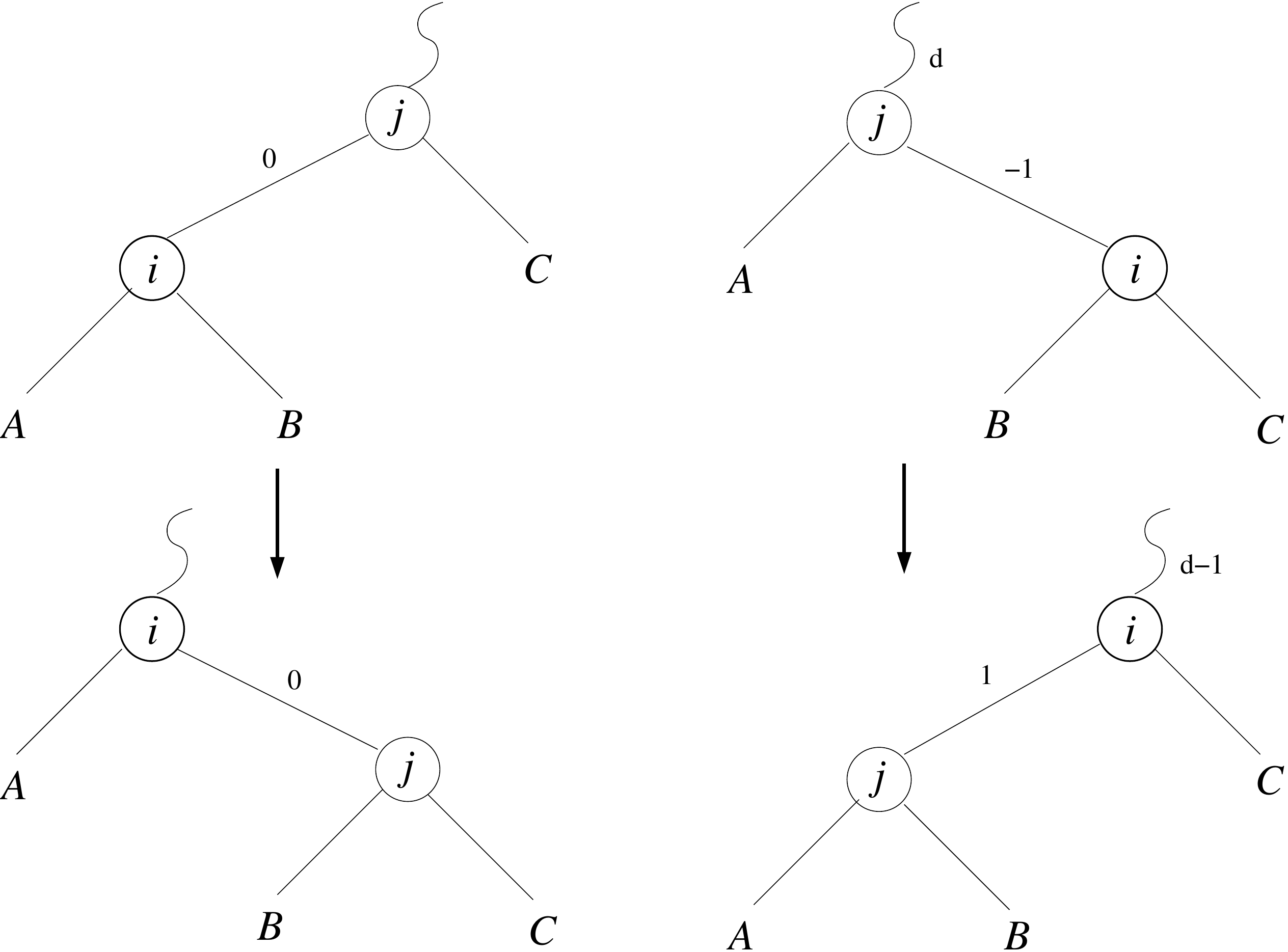}}
\caption{Rotations between $i$ and its parent $j$ when $m[i]$ decreases
by 1. The edges between any $x$ and its parent are labeled with $d[x] = m[x] -
m[Cparent(x)]$, if these change during the rotation. The $d[\cdot]$ values
have already been updated. On the left, when $i<j$, on the right, when $i>j$.}
\label{fig:mdecr}
\end{figure} 

Similarly, if $m[i]$ increases by 1, we must compare $i$ with its two children:
$(a)$ the difference with its left child cannot fall below 1 and $(b)$ the
difference with its right child cannot fall below 0. Otherwise we must carry 
out rotations as well, depicted in Figure~\ref{fig:mincr}. While it might seem 
that case $(b)$ can propagate rotations upwards (due to the $d-1$ at the
root), this is not the case because $d$ had just been increased as $m[i]$ grew
by 1. In case both $(a)$ and $(b)$ arise simultaneously, we must apply the
rotation corresponding to $(b)$ and then that of $(a)$. No further propagation
occurs. Again, universal tables can precompute all these updates.

\begin{figure}[bt]
\centerline{\includegraphics[width=0.67\textwidth]{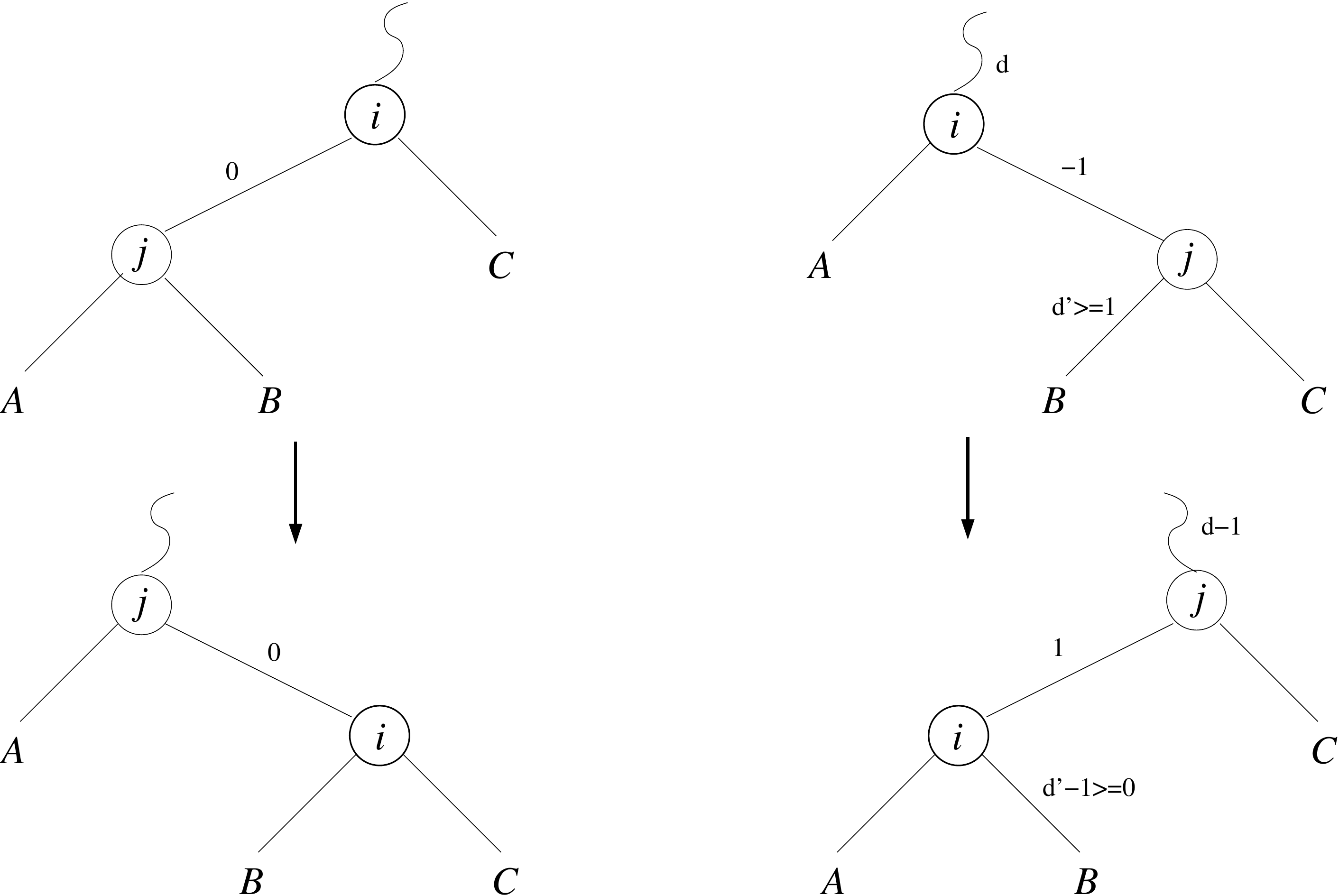}}
\caption{Rotations between $i$ and its children when $m[i]$ increases by 1. 
The edges between any $x$ and its parent are labeled with $d[x] = m[x] -
m[Cparent(x)]$, if these change during the rotation. The $d[\cdot]$ values
have already been updated. On the left (right), when the edge to the left 
(right) child becomes invalid after the change in $d[\cdot]$.}
\label{fig:mincr}
\end{figure} 

For changes in $e[i]$, the universal tables have precomputed the effect of 
carrying out all the changes in $m[i+1,k]$ , updating all the necessary 
$d[1,k]$ values and the Cartesian tree. This is equivalent to precomputing the 
effect of a sequence of $k-i$ successive changes in $m[\cdot]$.

Our array $d[1,k]$ distinguishes values between 0 and $k+2$. As the changes to
the structure of the Cartesian tree only depend on whether $d[i]$ is 0, 1, or
larger than 1, and all the updates to $d[i]$ are by $\pm 1$ per operation, we
have sufficient information in $d[\cdot]$ to correctly predict any change in 
the Cartesian tree shape for the next $k$ updates. We refresh table $d[\cdot]$ 
fast enough to ensure that no value of $d[\cdot]$ is used for more than $k$ 
updates without recomputing it, as then its imprecision could cause a flaw. 
We simply recompute cyclically the cells of $d[\cdot]$, one per update. 
That is, at the $i$-th 
update arriving at the node, we recompute the cell $i'=1+(i~\textrm{mod}~k)$, 
setting again $d[i'] = \min(k+2,m[i']-m[Cparent(i')])$; note that $Cparent(i')$
is computed from the Cartesian tree shape in constant time via table lookup. 
Note the values of $m[\cdot]$ are always up to date because we do not keep them
in explicit form but with $e[i-1]$ subtracted (and in turn $e$ is not 
maintained explicitly but via partial sums).

\subsection{Handling splits and merges}
\label{sec:splits}

In case of splits or merges of segments or internal range min-max tree nodes,
we must insert or delete children in a node. 
To maintain the range min-max tree dynamically, we use Fleischer's data
structure~\cite{Fleischer96}. This is an $(a,2b)$-tree (for $a \le 2b$)
storing $n$ numeric keys in the leaves, and each leaf is a bucket storing at 
most $2\log_a n$ keys. It supports constant-time insertion and deletion of a 
key once its location in a leaf is known.

Each leaf owns a {\em cursor}, which is a pointer to a tree node. This cursor
traverses the tree upwards, looking for nodes that should be split, moving one 
step per insertion received at the leaf. When the cursor reaches the root, the 
leaf has received at most $\log_a n$ insertions and thus it is split. Both new 
leaves are born with their cursor at their common parent.
%, and $\Order(\log b \cdot \log n)$ time traversal on the path
%from the root to a leaf.  
In addition some edges must be marked. Marks are considered when splitting
nodes (see Fleischer~\cite{Fleischer96} for details).
The insertion steps are as follows:
\begin{enumerate}
\item Insert the new key into the leaf $B$. Let $v$ be the current node where
the cursor of $B$ points.
\item If $v$ has more than $b$ children, split it into $v_1$ and $v_2$, and
unmark all edges leaving from those nodes.  If the parent of $v$ has more than 
$b$ children, mark the edges to $v_1$ and $v_2$.
\item If $v$ is not the root, set the cursor to the parent of $v$. Otherwise,
split $B$ into two halves, and let the cursor of both new buckets point to 
their common parent.
\end{enumerate}

To apply this to our data structure, let $a=\sqrt{\log n}, b = 2\sqrt{\log
n}$.
Then the height of the tree is $\Order(\log n/\log \log n)$, and each leaf
should store $\Theta(\log n/\log \log n)$ keys. Instead our structure stores
$\Theta(\log^2 n/\log \log n)$ bits in each leaf. If Fleischer's structure
handles $\Order(\log n)$-bit numbers, it turns out that the leaf size is the
same in both cases. The difference is that our insertions are bitwise, whereas
Fleischer's insertions are number-wise (that is, in packets of $\Order(\log
n)$
bits). Therefore we can use their same structure, yet the cursor will return
to the leaf $\Theta(\log n)$ times more frequently than necessary. Thus we
only split a leaf when the cursor returns to it {\em and} it has actually
exceeded size $2L$. This means leaves can actually reach size $L' = 2L +
\Order(\log n / \log\log n) = 2L(1+\Order(1/\log n))$, which is not a problem.
Marking and unmarking of children edges is easily handled in constant time by
storing a bit-vector of length $2b$ in each node.

%The traversal time becomes $\Order(\log n/\log \log n)$
%because we use the $f(i,d)/b(i,d)$ functions to determine which child to descend,
%instead of $\Order(\log b)$ time binary search.
Fleischer's update time is constant. Ours is $\Order(\sqrt{\log n})$ because,
if we split a node into two, we fully reconstruct all the values in those two 
nodes and their parent. This can be done in $\Order(k) = \Order(\sqrt{\log n})$
time, as the structure of Lemma~\ref{lem:partialsums}, the Cartesian trees, 
and the complete trees can be built in linear time. Nevertheless this time is 
dominated by the $\Order(\log n / \log\log n)$ cost of inserting a bit at the 
leaf.

Deletion of children of internal nodes may make the node arity fall below $a$.
This is handled as in Fleischer's structure, by deamortized global rebuilding.
This increases only the sublinear size of the range min-max tree; the leaves
are not affected. As a consequence, our tree arities are in the range
$1 \le k \le 4\sqrt{\log n}$.

Deletions at leaves, however, are handled as before, ensuring that they have
always between $L$ and $L'$ bits. This may cause some abrupt growth in their
length. The most extreme case arises when merging an underflowing leaf of $L-1$
bits with its sibling of length $L$. In this case the result is of size $2L-1$,
close to overflowing, it cannot be split, and the cursor may be far from the
root. This is not a problem, however, as we have still sufficient time before
the merged leaf reaches size $L'$.

%The original data structure~\cite{Fleischer96} uses the global rebuilding 
%technique to treat deletions. In our case this is replaced by the technique
%we use to handle changes in $\lceil \log n\rceil$.

\subsection{The final result}

We have obtained the following result.

\begin{lemma}
For a 0,1 vector of length $2n$, there exists a data structure
using $2n+\Order(n \log\log n/\log n)$ bits supporting
$\fwd$ and $\bwd$ in $\Order(\log n)$ time, 
and updates and all other queries except {\degree}, {\child}, and {\childrank},
in $\Order(\log n/\log\log n)$ time.
Alternatively, {\degree}, {\child}, {\childrank}, and updates can be handled
in $\Order(\log n)$ time.
\end{lemma}

The complexity of {\fwd} and {\bwd} is not completely satisfactory, as we
have reduced many operators to those. To achieve better complexities, we note 
that most operators that reduce to {\fwd} and {\bwd} actually reduce to the 
less general
operations {\findclose}, {\findopen}, and {\enclose} on parentheses. Those 
three operations can be supported in time $\Order(\log n / \log\log n)$ by
adapting the technique of Chan et al.~\cite{CHLS07}. They use a tree of similar
layout as ours: leaves storing $\Theta(\log^2 n / \log\log n)$ parentheses and
internal nodes of arity $k = \Theta(\sqrt{\log n})$, where 
Lemma~\ref{lem:partialsums} is used to store seven arrays of numbers 
recording information on matched and unmatched parentheses on the children.
Those are updated in constant time upon parenthesis insertions and deletions,
and are sufficient to support the three operations. They report $\Order(n)$
bits of space because they do not use a mechanism like the one we describe
in Section~\ref{sec:memory} for the leaves; otherwise their space would be $2n +
\Order(n\log\log n / \log n)$ as well. Note, on the other hand, that they do 
not achieve the times we offer for $\lca$ and related operations.

This completes the main result of this section, Theorem~\ref{th:dyn}.

\subsection{Updating whole subtrees}
\label{sec:subtrees}

We face now the problem of attaching and detaching whole subtrees. Now we
%\margin{Improve to $\log^{1+\epsilon} n$}
assume $\log n$ is fixed to some sufficiently large value, for example $\log n
= w$, the width of the systemwide pointers. Hence, no matter the size of the
trees, they use segments of the same length, and the times are a function of $w$
and not of the actual tree size.
%\margin{Think it's wrong.}

Now we cannot use Fleischer's data structure~\cite{Fleischer96}, because a
detached subtree could have dangling cursors pointing to the larger tree it 
belonged. As a result, the time complexity for insert or delete changes to 
$\Order(\log^{3/2} n / \log\log n)$.
To improve it, we change the degree of nodes in the range
min-max tree from $\Theta(\sqrt{\log n})$ to $\Theta(\log^\epsilon n)$ for
any fixed constant $\epsilon > 0$. This makes the complexity of $\ins$ and
$\del$ $\Order(\frac{1}{\epsilon}\log^{1+\epsilon} n/\log\log n) =
\Order(\log^{1+\epsilon} n)$, and multiplies all 
query time complexities by the constant $\Order(1/\epsilon)$.

First we consider attaching a tree $T_1$ to another tree $T_2$, that is,
$T_1$ becomes a subtree rooted at a node $v$ of $T_2$.  Here $v$ can be
either an internal node or a leaf.  Let $P_1[0,2n_1-1]$ and $P_2[0,2n_2-1]$
be the BP sequence of $T_1$ and $T_2$, respectively.  Then this attaching
operation corresponds to creating a new BP sequence
$P' = P_2[0,p] P_1[0,2n_1-1] P_2[p+1,2n_2-1]$ where $p$ and $p+1$ are positions
of parentheses for siblings of the root of $T_1$ in the new tree if $v$
is an internal node, or $p$ and $p+1$ are the positions for $v$ if $v$ is a leaf.

If $p$ and $p+1$ belong to the same segment, we cut the segment into two, say
$P_l = P[l,p]$ and $P_r = P[p+1,r]$.  If the length of $P_l$ ($P_r$)
is less than $L$, we concatenate it to the left (right) segment of it.
If its length exceeds $2L$, we split it into two.
We also update the upward paths from $P_l$ and $P_r$ to the root of the
range min-max tree for $T_2$ to reflect the changes done at the leaves.

Now we merge the range min-max trees for $T_1$ and $T_2$ as follows.
Let $h_1$ be the height of the range min-max tree of $T_1$,
and $h_2$ be the height of the {\lca}, say $v$, between $P_l$ and $P_r$
in the range min-max tree of $T_2$.
If $h_2>h_1$ then can simply concatenate the root of $T_1$ at the right of
the ancestor of $P_l$ of height $h_1$, then split the node if it has
overflowed, and finish.

If $h_2 \le h_1$, we divide $v$ into $v_l$ and $v_r$, so that the rightmost
child of $v_l$ is an ancestor of $P_l$ and the leftmost child of $v_r$ is
an ancestor of $P_r$. We do not yet care about $v_l$ or $v_r$ being too
small. We repeat the process on the parent of $v$ until reaching the height
$h_2=h_1+1$. Let us call $u$ the ancestor where this height is reached (we 
leave for later the case where we split the root of $T_2$ without reaching 
the height $h_1+1$).

Now we add $T_1$ as a child of $u$, between the child ancestor of $P_l$ and
that ancestor of $P_r$. All the leaves have the same depth, but the
ancestors of $P_l$ and of $P_r$ at heights $h_2$ to $h_1$ might be underfull
as we have cut them arbitrarily. We glue the ancestor of height $h$ of $P_l$ 
with the leftmost node of height $h$ of $T_1$, and that of $P_r$ with the 
rightmost node of $T_1$, for all $h_2 \le h \le h1$. Now there are no 
underfull nodes, but they can have overflowed. We verify the node sizes in
both paths, from height $h=h_2$ to $h_1+1$, splitting them as necessary. At
height $h_2$ the node can be split into two, adding another child to its
parent, which can thus be split into three, adding in turn two children to its
parent, but from there on nodes can only be split into three and add two more
children to their parent. Hence the overall process of fixing arities takes
time $\Order(\frac{1}{\epsilon}\log^{1+\epsilon} n/\log\log n)$.

If node $u$ does not exist, then $T_1$ is not shorter than $T_2$. In this case 
we have divided $T_2$ into a left and right part. Let $h_2$ be the height of
$T_2$. We attach the left part of $T_2$ to the 
leftmost node of height $h_2$ in $T_1$, and the right part of $T_2$ to the
rightmost node of height $h_2$ in $T_1$. Then we fix arities in both paths
analogously as before.

Detaching is analogous as well. After splitting the leftmost and rightmost
leaves of the area to be detached, let $P_l$ and $P_r$ the leaves of $T$ 
preceding and following the leaves that will be detached. We split the
ancestors of $P_l$ and $P_r$ until reaching their $\lca$, let it be $v$.
Then we can form a new tree with the detached part and remove it from the
original tree $T$. Again, the paths from $P_l$ and $P_r$ to $v$ may contain
underfull nodes. But now $P_l$ and $P_r$ are consecutive leaves, so we can
merge their ancestor paths up to $v$ and then split as necessary. 

Similarly, the leftmost and rightmost path of the detached tree may contain
underfull nodes. We merge each node of the leftmost (rightmost) path with its
right (left) sibling, and then split if necessary. The root may contain as
few as two children. Overall the process takes
$\Order(\log^{1+\epsilon} n)$ time.

\section{Improving dynamic compressed sequences}
\label{sec:sequences}

The techniques we have developed along the paper are of independent interest.
We illustrate this point by improving the best current results on sequences of 
numbers with $sum$ and $search$ operations, dynamic compressed bitmaps, and 
their many byproducts.

\subsection{Codes, Numbers, and Partial Sums}

We prove now Lemma~\ref{lem:dynpartialsums} on sequences of codes and
partial sums, this way improving previous results by M\"akinen and Navarro 
\cite{MN08} and matching lower bounds \cite{PD06}.

Section~\ref{sec:faster} shows how to maintain a dynamic bitmap $P$
supporting various operations in time $\Order(\log n / \log\log n)$, 
including insertion and deletion of bits (parentheses in $P$). This bitmap $P$
will now be the concatenation of the (possibly variable-length) codes 
$x_i$. We will ensure that each leaf contains a sequence of whole codes 
(no code is split at a leaf boundary). As these are of $\Order(\log n)$ bits, 
we only need to slightly adjust the lower limit $L$ to enforce this: After 
splitting a leaf of length $2L$, one of the two new leaves might be of size 
$L-\Order(\log n)$.

We process a leaf by chunks of $b = \frac{1}{2}\log n$ bits: A universal table 
(easily computable in $\Order(\sqrt{n}\,\polylog(n))$ time and space)
can tell us how many whole codes are there in the next $b$ bits, how much 
their $f(\cdot)$ values add up to, and where the last complete code ends (assuming we 
start reading at a code boundary). Note that the first code could be 
longer than $b$, in which case the table lets us advance zero positions. In 
this case we decode the next code directly. Thus in constant time (at most two 
table accesses plus one direct decoding) we advance in the traversal by at 
least $b$ bits. If we surpass the desired position with the table we reprocess 
the last $\Order(\log n)$ codes using a second table that advances by chunks
of $\Order(\sqrt{\log n})$ bits, and finally process the last
$\Order(\sqrt{\log n})$ codes directly. Thus in time $\Order(\log n /
\log\log n)$ we can access a given code in a leaf (and subsequent ones in 
constant time each), sum the $f(\cdot)$ values up to some position, and find the 
position where a given sum $s$ is exceeded. 
We can also easily modify a code or insert/delete codes, by shifting all
the other codes of the leaf in time $\Order(\log n / \log\log n)$.

In internal nodes of the range min-max tree we will use the structure of 
Lemma~\ref{lem:partialsums} to maintain the number of codes stored below the 
subtree of 
each child of the node. This allows determining in constant time the child to 
follow when looking for any code $x_i$, thus access to any codes $x_i \ldots
x_j$ is supported in time $\Order(\log n / \log\log n + j-i)$. 

When a code is inserted/deleted at a leaf, we must increment/decrement 
the number of codes in the subtree of the ancestors up to the root; this is
supported in constant time by Lemma~\ref{lem:partialsums}. Splits and merges
can be caused by indels and by updates. They force the recomputation of their
whole parent node, and Fleischer's technique is used to ensure a constant
number of splits/merges per update. Note we are inserting not individual bits
but whole codes of $\Order(\log n)$ bits. This can easily be done, but now
$\Order(\log n / \log\log n)$ insertions/updates can double the size of a
leaf, and thus we must consider splitting the leaf every time the cursor
returns to it (as in the original Fleischer's proposal, not every $\log n$ 
times as when inserting parentheses), and we
must advance the cursor upon insertions {\em and} updates.

Also, we must allow leaves of sizes between $L$ and $L'=3L$ (but still split
them as soon as they exceed $2L$ bits). In this way, after a merge produces a
leaf of size $2L-1$, we still have time to carry out $L = \Order(\log n /
\log\log n)$ further insertions before the cursor reaches the root and splits
the leaf. (Recall that if the merge produces a leaf larger than $2L$ we can
immediately split it, so $2L$ is the worst case we must handle.)

For supporting $sum$ and $search$ we also maintain at each node the sum of the
$f(\cdot)$ values of the codes stored in the subtree of each child. Then we can
determine in constant time the child to follow for $search$, and the sum of
previous subtrees for $sum$. However, insertions, deletions and updates must 
alter the upward sums only by $\Order(\log n)$ so that the change can be 
supported by Lemma~\ref{lem:partialsums} within the internal nodes in constant 
time.

\subsection{Dynamic bitmaps}
\label{sec:dynbitmaps}

Apart from its general interest, handling a dynamic bitmap in compressed 
form is useful for
maintaining satellite data for a sample of the tree nodes. A dynamic bitmap 
$B$ could mark which nodes are sampled, so if the sampling is sparse enough we 
would like $B$ to be compressed. A \rank\ on this bitmap would give the 
position in a dynamic array where the satellite information 
for the sampled nodes would be stored. This bitmap would be accessed by 
preorder (\preorderrank) on the dynamic tree. That is, node $v$ is sampled iff 
$B[\preorderrank(v)]=1$, and if so, its data is at position 
$\rank_1(B,\preorderrank(v))$ in the dynamic array of satellite data.
When a tree node is inserted or deleted, we need to insert/delete its 
corresponding bit in $B$.

In the following we prove the next lemma, which improves and indeed simplifies
previous results \cite{CHLS07,MN08}; then we explore several byproducts.%
\footnote{Very recently, He and Munro \cite{HM10} achieved results similar to
Lemma~\ref{lem:bitmap}, Theorem~\ref{thm:seqs} and Theorem~\ref{thm:fmindex}
(excluding attachment and detachment of sequences), independently and with a
different technique. The results differ in the space redundancy, which
in their case is $\Order(1/\sqrt{\log n})$ times the sequence, versus our
$\Order(\log\log n / \log n)$.}

\begin{lemma} \label{lem:bitmap}
We can store any bitmap $B[0,n-1]$ within $nH_0(B)+\Order(n\log\log n/\log n)$
bits of space, while supporting the operations $\rank$, $\select$, $\ins$,
and $\del$, all in time $\Order(\log n / \log\log n)$. We can also support
attachment and detachment of contiguous bitmaps within time 
$\Order(\log^{1+\epsilon} n)$ for any constant $\epsilon>0$, yet now $\log n$
is a maximum fixed value across all the operations.
\end{lemma}

To achieve zero-order entropy space, we use Raman et al.'s 
$(c,o)$ encoding \cite{RRR02}: The bits are grouped into small chunks of 
$b=\frac{\log n}{2}$ bits, and each chunk is represented by two components: 
the {\em class} $c_i$, which is the number of bits set, and the {\em offset}
$o_i$, which is an identifier of that chunk within those of the same class. 
Raman et al.\ show that, while the $|c_i|$ lengths add up to 
$\Order(n\log\log n/\log n)$ extra bits, the $|o_i| = \left\lceil \log {b
\choose c_i} \right\rceil$ components add up to $nH_0(B)+\Order(n/\log n)$ bits.

We plan to store whole chunks in leaves of the range min-max tree. A problem is
that the insertion or even deletion of a single bit in Raman et al.'s 
representation can up to double the size of the compressed representation of 
the segment, because it can change all the alignments. This occurs for example 
when moving from $0^b~1^b~0^b~1^b\ldots $ to $10^{b-1}~01^{b-1}~10^{b-1}~01^{b-1} 
\ldots$, where we switch from all $c_i=0/b$ and $|o_i|=0$, to all
$c_i=1/b-1$, and $|o_i|=\lceil \log b \rceil$. This problem can be dealt
with (laboriously) on binary trees \cite{MN08,GN08}, but not on our $k$-ary 
tree, because Fleischer's scheme does not allow leaves being partitioned often
enough.

We propose a different solution that ensures that an insertion cannot make the
leaf's physical size grow by more than $\Order(\log n)$ bits. Instead of
using the same $b$ value for all the chunks, we allow any $1 \le b_i \le b$.
Thus each chunk is represented by a triple $(b_i,c_i,o_i)$, where $o_i$ is
the offset of this chunk among those of length $b_i$ having $c_i$ bits set.
To ensure $\Order(n\log\log n/\log n)$ space overhead over the entropy, we
state the invariant that any two consecutive chunks $i$, $i+1$ must satisfy
$b_i + b_{i+1} > b$. Thus there are $\Order(n/b)$ chunks and the overhead of
the $b_i$ and $c_i$ components, representing each with $\lceil \log(b+1)
\rceil$ bits, is
$\Order(n\log b/b)$. It is also easy to see that the inequality \cite{Pag01}
$\sum |o_i| = \sum \lceil \log {b_i \choose c_i} \rceil = 
\log \Pi {b_i \choose c_i} + \Order(n/\log n) \le
\log {n \choose m} + \Order(n/\log n) = nH_0(B) + \Order(n/\log n)$ holds,
where $m$ is the number of 1s in the bitmap.

To maintain the invariant, the insertion of a bit us processed as follows.
We first identify the chunk $(b_i,c_i,o_i)$ where the bit must be 
inserted, and compute its new description $(b_i',c_i',o_i')$. If $b_i' >
b$, we split the chunk into two, $(b_l,c_l,o_l)$ and $(b_r,c_r,o_r)$, 
for $b_l,b_r = b_i'/2 \pm 1$. Now we check left and 
right neighbors $(b_{i-1},c_{i-1},o_{i-1})$ and $(b_{i+1},c_{i+1},o_{i+1})$
to ensure the invariant on consecutive chunks holds. If $b_{i-1}+b_l \le b$
we merge these two chunks, and if $b_r+b_{i+1} \le b$ we merge these
two as well. Merging is done in constant time by obtaining the plain bitmaps, 
concatenating them, and reencoding them, using universal tables (which we must
have for all $1 \le b_i \le b$). Deletion of a bit is analogous; we remove the
bit and then consider the conditions $b_{i-1} + b_i' \le b$ and
$b_i'+b_{i+1} \le b$. It is easy to see that no insertion/deletion can increase
the encoding by more than $\Order(\log n)$ bits.

Now let us consider {\em codes} $x_i = (b_i,c_i,o_i)$. These are clearly
constant-time self-delimiting and $|x_i| = \Order(\log n)$, so we can
directly use Lemma~\ref{lem:dynpartialsums} to store them in a range min-max
tree within $n'+\Order(n'\log\log n' / \log n')$ bits, where
$n' = nH_0(B)+\Order(n\log\log n/\log n)$ is the size of our compressed 
representation. Since $n' \le n + \Order(n\log\log n/\log n)$, we have
$\Order(n'\log\log n'/\log n') = \Order(n\log\log n/ \log n)$ and
the overall space is as promised in the lemma. We must only take care of
checking the invariant on consecutive chunks when merging leaves, which takes
constant time.

Now we use the $sum$/$search$ capabilities of Lemma~\ref{lem:dynpartialsums}.
Let $f_b(b_i,c_i,o_i) = b_i$ and $f_c(b_i,c_i,o_i) = c_i$. As both are always 
$\Order(\log n)$, we can have $sum$/$search$ support on them. With $search$ on 
$f_b$ we can reach the code containing the $j$th bit of the original sequence, 
which is key for accessing an arbitrary bit. For supporting $\rank$ we need to 
descend using $search$ on $f_b$, and accumulate the $sum$ on the $f_c$ values 
of the left siblings as we descend. For supporting $\select$ we descend using 
$search$ on $f_c$, and accumulate the $sum$ on the $f_b$ values. Finally, for
insertions and deletions of bits we first access the proper position, and then
implement the operation via a constant number of updates, insertions, and 
deletions of codes (for updating, splitting, and merging our triplets). 
Thus we implement all the operations within time $\Order(\log n / \log\log n)$.

We can also support attachment and detachment of contiguous bitmaps, by
applying essentially the same techniques developed in
Section~\ref{sec:subtrees}. We can have a bitmap $B'[0,n'-1]$ and 
insert it between $B[i]$ and $B[i+1]$, or we can detach any $B[i,j]$ from $B$
and convert it into a separate bitmap that can be handled independently. The
complications that arise when cutting the compressed segments at arbitrary
positions are easily handled by splitting codes. Zero-order compression is 
retained as it is due to the sum of the local entropies of the chunks, which 
are preserved (small resulting segments after the splits are merged as usual). 

\subsection{Sequences and Text Indices}

We now aim at maintaining a sequence $S[0,n-1]$ of symbols over an alphabet 
$[1,\sigma]$, so that we can insert and delete symbols, and also compute
symbol $\rank_c(S,i)$ and $\select_c(S,i)$, for $1\le c\le \sigma$. This has
in particular applications to labeled trees: We can store the sequence $S$
of the labels of a tree in preorder, so that $S[\preorderrank(i)]$ is the
label of node $i$. Insertions and deletions of nodes must be accompanied with 
insertions and deletions of their labels at the corresponding preorder
positions, and this can be extended to attaching and detaching subtrees. 
Then we not only have easy access to the label of each node, but can also
use $\rank$ and $\select$ on $S$ to find the $r$-th descendant node labeled 
$c$, or compute the number of descendants labeled $c$. If the balanced 
parentheses represent the tree in DFUDS format~\cite{BenDemMunRamRamRao05}, 
we can instead find the first child of a node labeled $c$ using $\select$.

We divide the sequence into chunks of maximum size $b=\frac{1}{2}\log_\sigma n$
symbols and store them using an extension of the $(c_i,o_i)$ encoding for 
sequences \cite{FerManMakNav07}. Here $c_i = (c_i^1,\ldots,c_i^\sigma)$,
where $c_i^a$ is the number of occurrences of character $a$ in the chunk.
For this code to be of length $\Order(\log n)$ we need 
$\sigma=\Order(\log n / \log\log n)$; more stringent conditions will arise 
later. To this code we add the $b_i$ component as in 
Section~\ref{sec:dynbitmaps}. This takes 
$nH_0(S)+\Order(\frac{n \sigma \log\log n}{\log n})$ bits of space.
In the range min-max tree nodes, which we again assume to hold 
$\Theta(\log^\epsilon n)$ children for some constant $0<\epsilon<1$, instead of
a single $f_c$ function as in Section~\ref{sec:dynbitmaps}, we must store one 
$f_a$ function for each $a \in [1,\sigma]$, requiring extra space 
$\Order(\frac{n\sigma\log\log n}{\log n})$.
Symbol \rank\ and \select\ are easily carried out by considering the proper
$f_a$ function. Insertion and deletion of symbol $a$ is carried out in the 
compressed sequence as before, and only $f_b$ and $f_a$ sums must be 
incremented/decremented along the path to the root.

In case a leaf node splits or merges, we must rebuild the partial sums for
all the $\sigma$ functions $f_a$ (and the single function $f_b$) of a node, 
which requires $\Order(\sigma \log^\epsilon n)$ time. In 
Section~\ref{sec:splits} we have shown how to limit the number of
splits/merges to one per operation, thus we can handle all the operations
within $\Order(\log n / \log\log n)$ time as long as 
$\sigma = \Order(\log^{1-\epsilon} n / \log\log n)$. 
%The extra space with this limited $\sigma$ is 
%$\Order(\frac{1}{\epsilon}n/\log^\epsilon n)$ bits.
This, again, greatly simplifies the solution by Gonz\'alez and Navarro
\cite{GN08}, which used a collection of searchable partial sums with indels. 

Up to here, the result is useful for small alphabets only. Gonz\'alez and 
Navarro \cite{GN08} handle larger alphabets by using a multiary wavelet 
tree (Section~\ref{sec:mwt}). Recall this is a complete $r$-ary tree of height
$h = \lceil \log_r \sigma \rceil$ that stores a string over alphabet $[1,r]$ 
at each node. It solves all the operations (including insertions and deletions)
by $h$ applications of the analogous operation on the sequences over alphabet 
$[1,r]$. 

Now we set $r = \log^{1-\epsilon} n/\log\log n$, and use the small-alphabet 
solution to handle the sequences stored at the wavelet tree nodes. The height 
of the wavelet tree is 
$h = \Order\left(1+\frac{\log\sigma}{(1-\epsilon)\log\log n}\right)$.
The zero-order entropies of the small-alphabet sequences add up to that of
the original sequence and the redundancies add up to 
$\Order\left(\frac{n\log\sigma}{(1-\epsilon)\log^\epsilon n \log\log n}\right)$.
The operations time is
$\Order\left(\frac{\log n}{\log\log n}\left(1 + \frac{\log\sigma}{(1-\epsilon)\log\log n}\right)\right)$. 
By slightly altering
$\epsilon$, we achieve the first part of Theorem~\ref{thm:seqs}, where the 
$\Order(\sigma\log^\epsilon n)$ term owes to representing the wavelet tree 
itself, which has $\Order(\sigma/r)$ nodes.

For the second part, the arity of the nodes fixed to $\Theta(\log^\epsilon n)$ allows us attach
and detach substrings in time $\Order(r\log^{1+\epsilon} n)$ on a sequence with
alphabet size $r$. This has to be carried out on each of the $\Order(\sigma/r)$
wavelet tree nodes, reaching overall complexity 
$\Order(\sigma\log^{1+\epsilon} n)$.

The theorem has immediate application to the handling of compressed
dynamic text collections, construction of compressed static text collections
within compressed space, and construction of the Burrows-Wheeler transform
(BWT) within compressed space. We state them here for completeness; for their
derivation refer to the original articles \cite{MN08,GN08}. 

The first result refers to maintaining a collection of texts within high-entropy
space, so that one can perform searches and also insert and delete texts. Here 
$H_h$ refers to the $h$-th order empirical entropy of a sequence, see 
e.g.~Manzini \cite{Man01}. We use sampling step $\log_\sigma n \log\log n$
to achieve it.

\begin{theorem} \label{thm:fmindex}
There exists a data structure for handling a collection $\mathcal{C}$ of
texts over an alphabet $[1,\sigma]$ within size
$nH_h(\mathcal{C})+o(n\log\sigma)+\Order(\sigma^{h+1}\log n + m\log n+w)$
bits, simultaneously for all $h$. Here $n$ is the length of the concatenation
of $m$ texts, $\mathcal{C}=0\ T_1 0\ T_2 \cdots$ $0\ T_m$, and we assume that
$\sigma=o(n)$ is the alphabet size and $w=\Omega(\log n)$ is the machine word
size under the RAM model. The structure supports counting of the occurrences
of a pattern $P$ in 
$\Order(|P|\frac{\log n}{\log\log n} (1+\frac{\log \sigma}{\log\log n}))$ time, and
inserting and deleting a text $T$ in
$\Order(\log n + |T|\frac{\log n}{\log\log n} (1+\frac{\log \sigma}{\log\log n}))$ time. 
After counting, any occurrence can be located in time 
$\Order(\frac{\log^2 n}{\log\log n}(1 + \frac{\log\log n}{\log\sigma}))$. Any
substring of length $\ell$ from any $T$ in the collection can be displayed in
time $\Order(\frac{\log^2 n}{\log\log n} (1+ \frac{\log\log n}{\log\sigma}) +
\ell\frac{\log n}{\log\log n}(1+\frac{\log\sigma}{\log\log n}))$.
For $h \le (\alpha \log_\sigma n)-1$, for any constant $0<\alpha<1$, the space
complexity simplifies to $nH_h(\mathcal{C})+o(n\log\sigma)+\Order(m\log n+w)$ bits.
\end{theorem}

The second result refers to the construction of the most succinct self-index
for text within the same asymptotic space required by the final structure.
This is tightly related to the construction of the BWT, which has many
applications.

\begin{theorem}
The Alphabet-Friendly FM-index \cite{FerManMakNav07}, as well as the 
BWT \cite{BWT}, of a text $T[0,n-1]$ over an alphabet of size
$\sigma$, can be built using $nH_h(T)+o(n\log\sigma)$ bits, simultaneously for
all $h \le (\alpha \log_\sigma n)-1$ and any constant $0<\alpha<1$, in time
$\Order(n\frac{\log n}{\log\log n} (1+\frac{\log \sigma}{\log\log n}))$.
\end{theorem}

On polylog-sized alphabets, we build the BWT in $o(n\log n)$ time. Even on a 
large alphabet $\sigma=\Theta(n)$, we build the BWT in $o(n\log^2 n)$ time.
This slashes by a $\log\log n$ factor the corresponding previous result
\cite{GN08}. Other previous results that focus in using little space are as
follows. Okanohara and Sadakane \cite{OS09} achieved optimal $\Order(n)$ 
construction time with $\Order(n \log \sigma \log\log_\sigma n)$ bits of extra 
space (apart from the $n\log\sigma$ bits of the sequence). Hon et 
al.~\cite{HonSadSun09}
achieve $\Order(n\log\log\sigma)$ time and $\Order(n\log\sigma)$ bits
of extra space.
Ours is the fastest construction within compressed space.

\section{Concluding remarks}\label{sec:conclusion}

We have proposed flexible and powerful data structures for the succinct 
representation of ordinal trees. For the static case, all the known operations 
are done in constant time using $2n+ \Order(n /\polylog(n))$ bits of space, for
a tree of $n$ nodes and a polylog of any degree. This significantly improves 
upon the redundancy of previous representations.
The core of the idea is the range min-max tree.
This simple data structure reduces all of
the operations to a handful of primitives, which run in constant time on
polylog-sized subtrees. It can be used in standalone form to obtain a simple 
and practical implementation that achieves $\Order(\log n)$ time for all the 
operations. We then show how constant time can be achieved by using the range 
min-max tree as a building block for handling larger trees.

The simple implementation using one range min-max tree has actually been 
implemented and compared with the state of the art over several real-life 
trees~\cite{ACNS10}. It has been
shown that it is by far the smallest and fastest representation in most cases,
as well as the one with widest coverage of operations. It requires around
2.37 bits per node and carries out most operations within the microsecond on a
standard PC.

For the dynamic case, there have been no data structures supporting several of
the usual tree operations. The data structures of this paper support all of the 
operations, including node insertion and deletion, in $\Order(\log n)$ time, 
and a variant supports most of them in $\Order(\log n/\log\log n)$ time, which
is optimal in the dynamic case even for a very reduced set of operations.
They are based on dynamic range min-max trees, and especially the former is
extremely simple and implementable. We expect a performance similar to that of 
the static version in practice. Their flexibility is illustrated by the fact that we can 
support much more complex operations, such as attaching and detaching whole 
subtrees.

Our work contains several ideas of independent interest. An immediate 
application to storing a dynamic sequence of numbers supporting operations
$sum$ and $search$ achieves optimal time $\Order(\log n / \log\log n)$.
Another application is the storage of dynamic compressed sequences achieving
zero-order entropy space and improving the redundancy of previous work. 
It also improves the times 
for the operations, achieving the optimal $\Order(\log n / \log\log n)$ for
polylog-sized alphabets. This in turn has several applications to compressed
text indexing.

P\v{a}tra\c{s}cu and Viola have recently shown that $n+n/w^{\Theta(c)}$ bits are
necessary to compute $\rank$ or $\select$ on bitmaps in time $\Order(t)$ in 
the worst case \cite{PV10}. This lower bound holds also in the subclass of
balanced bitmaps\footnote{M. P\v{a}tra\c{s}cu, personal communication.}
(i.e., those corresponding to balanced parenthesis sequences), which makes our 
redundancy on static trees optimal as well, at least for some of the 
operations: Since {\rank} or {\select} can be obtained from any of the
operations {\depth}, {\preorderrank}, {\postorderrank}, {\preorderselect},
{\postorderselect}, any balanced parentheses representation supporting
any of these operations in time $\Order(c)$ requires $2n+2n/w^{\Theta(c)}$
bits of space. Still, it would be good to show a lower bound for the 
more fundamental set of operations {\findopen}, {\findclose}, and {\enclose}.

On the other hand, the complexity $\Order(\log n/\log\log n)$ is known to be 
optimal for several basic dynamic tree operations, but not for all. It is also 
not clear if the redundancy $\Order(n/r)$ achieved for the dynamic trees, 
$r=\log n$ for the simpler structure and $r=\log\log n/\log n$ for the more 
complex one, is optimal to achieve the corresponding $\Order(r)$ operation 
times. 
Finally, it would be good to achieve $\Order(\log n / \log\log n)$ time for 
all the operations or prove it impossible.

\section*{Acknowledgments}

We thank Mihai P\v{a}tra\c{s}cu for confirming us the construction cost of his
aB-tree and rank/select data structure \cite{Pat08}. And we thank him again
for confirming us that the lower bound \cite{PV10} holds for balanced
sequences.

\bibliographystyle{plain}
\bibliography{compression2,string,mypapers,kai,grossi}

\begin{thebibliography}{10}

\bibitem{Arr08}
D.~Arroyuelo.
\newblock An improved succinct representation for dynamic $k$-ary trees.
\newblock In {\em Proc. 19th Annual Symposium on Combinatorial Pattern Matching
  (CPM)}, LNCS 5029, pages 277--289, 2008.

\bibitem{ACNS10}
D.~Arroyuelo, R.~C{\'a}novas, G.~Navarro, and K.~Sadakane.
\newblock Succinct trees in practice.
\newblock In {\em Proc. 11th Workshop on Algorithm Engineering and Experiments
  (ALENEX)}, pages 84--97. SIAM Press, 2010.

\bibitem{BMHR07}
J.~Barbay, J.~I. Munro, M.~He, and S.~S. Rao.
\newblock Succinct indexes for strings, binary relations and multi-labeled
  trees.
\newblock In {\em Proc. 18th Annual ACM-SIAM Symposium on Discrete Algorithms
  (SODA)}, pages 680--689, 2007.

\bibitem{BenFar00}
M.~Bender and M.~Farach-Colton.
\newblock The {LCA} problem revisited.
\newblock In {\em Proc. 4th Latin American Symposium on Theoretical Informatics
  (LATIN)}, LNCS 1776, pages 88--94, 2000.

\bibitem{BenFar04}
M.~Bender and M.~Farach-Colton.
\newblock The level ancestor problem simplified.
\newblock {\em Theoretical Computer Science}, 321(1):5--12, 2004.

\bibitem{BenDemMunRamRamRao05}
D.~Benoit, E.~D. Demaine, J.~I. Munro, R.~Raman, V.~Raman, and S.~S. Rao.
\newblock Representing trees of higher degree.
\newblock {\em Algorithmica}, 43(4):275--292, 2005.

\bibitem{BWT}
M.~Burrows and D.J. Wheeler.
\newblock A block sorting data compression algorithm.
\newblock Technical report, Digital Systems Research Center, 1994.

\bibitem{CHLS07}
H.-L. Chan, W.-K. Hon, T.-W. Lam, and K.~Sadakane.
\newblock Compressed indexes for dynamic text collections.
\newblock {\em {ACM Transactions on Algorithms}}, 3(2):article 21, 2007.

\bibitem{ChiLinLu05}
Y.-T. Chiang, C.-C. Lin, and H.-I. Lu.
\newblock Orderly spanning trees with applications.
\newblock {\em SIAM Journal on Computing}, 34(4):924--945, 2005.

\bibitem{DRR06}
O.~Delpratt, N.~Rahman, and R.~Raman.
\newblock Engineering the {LOUDS} succinct tree representation.
\newblock In {\em Proc. 5th Workshop on Efficient and Experimental Algorithms
  (WEA)}, pages 134--145. LNCS 4007, 2006.

\bibitem{FarMun08}
A.~Farzan and J.~I. Munro.
\newblock A uniform approach towards succinct representation of trees.
\newblock In {\em Proc. 11th Scandinavian Workshop on Algorithm Theory (SWAT)},
  LNCS 5124, pages 173--184, 2008.

\bibitem{FLMM05}
P.~Ferragina, F.~Luccio, G.~Manzini, and S.~Muthukrishnan.
\newblock Structuring labeled trees for optimal succinctness, and beyond.
\newblock In {\em Proc. 46th IEEE Annual Symposium on Foundations of Computer
  Science (FOCS)}, pages 184--196, 2005.

\bibitem{FerManMakNav07}
P.~Ferragina, G.~Manzini, V.~M{\"a}kinen, and G.~Navarro.
\newblock {Compressed Representations of Sequences and Full-Text Indexes}.
\newblock {\em {ACM Transactions on Algorithms}}, 3(2):No. 20, 2007.

\bibitem{Fis10}
J.~Fischer.
\newblock Optimal succinctness for range minimum queries.
\newblock In {\em Proc. 9th Symposium on Latin American Theoretical Informatics
  (LATIN)}, LNCS 6034, pages 158--169, 2010.

\bibitem{FH07}
J.~Fischer and V.~Heun.
\newblock A new succinct representation of {RMQ}-information and improvements
  in the enhanced suffix array.
\newblock In {\em Proc. 1st International Symposium on Combinatorics,
  Algorithms, Probabilistic and Experimental Methodologies (ESCAPE)}, LNCS
  4614, pages 459--470, 2007.

\bibitem{Fleischer96}
R.~Fleischer.
\newblock A simple balanced search tree with {O(1)} worst-case update time.
\newblock {\em International Journal of Foundations of Computer Science},
  7(2):137--149, 1996.

\bibitem{FreSak89}
M.~Fredman and M.~Saks.
\newblock {The Cell Probe Complexity of Dynamic Data Structures}.
\newblock In {\em Proc. 21st Annual ACM Symposium on Theory of Computing
  (STOC)}, pages 345--354, 1989.

\bibitem{FW93}
M.~Fredman and D.~Willard.
\newblock Surpassing the information theoretic bound with fusion trees.
\newblock {\em Journal of Computer and Systems Science}, 47(3):424--436, 1993.

\bibitem{GearyRRR04}
R.~F. Geary, N.~Rahman, R.~Raman, and V.~Raman.
\newblock A simple optimal representation for balanced parentheses.
\newblock In {\em Proc. 15th Annual Symposium on Combinatorial Pattern Matching
  (CPM)}, LNCS 3109, pages 159--172, 2004.

\bibitem{GRR04}
R.~F. Geary, R.~Raman, and V.~Raman.
\newblock Succinct ordinal trees with level-ancestor queries.
\newblock In {\em Proc. 15th Annual ACM-SIAM Symposium on Discrete Algorithms
  (SODA)}, pages 1--10, 2004.

\bibitem{GGGRR07}
A.~Golynski, R.~Grossi, A.~Gupta, R.~Raman, and S.~S. Rao.
\newblock On the size of succinct indices.
\newblock In {\em Proc. 15th Annual European Symposium on Algorithms (ESA)},
  pages 371--382. LNCS 4698, 2007.

\bibitem{GN08}
R.~Gonz{\'a}lez and G.~Navarro.
\newblock Rank/select on dynamic compressed sequences and applications.
\newblock {\em Theoretical Computer Science}, 410:4414--4422, 2008.

\bibitem{GroGupVit03a}
R.~Grossi, A.~Gupta, and J.~S. Vitter.
\newblock {High-Order Entropy-Compressed Text Indexes}.
\newblock In {\em Proc. 14th Annual ACM-SIAM Symposium on Discrete Algorithms
  (SODA)}, pages 841--850, 2003.

\bibitem{HM10}
M.~He and I.~Munro.
\newblock Succinct representations of dynamic strings.
\newblock In {\em Proc. 17th International Symposium on String Processing and
  Information Retrieval (SPIRE)}, LNCS, 2010.
\newblock To appear.

\bibitem{HMR07}
M.~He, J.~I. Munro, and S.~S. Rao.
\newblock Succinct ordinal trees based on tree covering.
\newblock In {\em Proc. 34th International Colloquium on Automata, Languages
  and Programming (ICALP)}, LNCS 4596, pages 509--520, 2007.

\bibitem{HonSadSun09}
W.~K. Hon, K.~Sadakane, and W.~K. Sung.
\newblock {Breaking a Time-and-Space Barrier in Constructing Full-Text
  Indices}.
\newblock {\em {SIAM Journal on Computing}}, 38(6):2162--2178, 2009.

\bibitem{Jacobson89}
G.~Jacobson.
\newblock Space-efficient static trees and graphs.
\newblock In {\em Proc. 30th IEEE Annual Symposium on Foundations of Computer
  Science (FOCS)}, pages 549--554, 1989.

\bibitem{JanSadSun07}
J.~Jansson, K.~Sadakane, and W.-K. Sung.
\newblock Ultra-succinct representation of ordered trees.
\newblock In {\em Proc. 18th Annual ACM-SIAM Symposium on Discrete Algorithms
  (SODA)}, pages 575--584, 2007.

\bibitem{LuYeh07}
H.-I. Lu and C.-C. Yeh.
\newblock Balanced parentheses strike back.
\newblock {\em ACM Transactions on Algorithms}, 4(3):article 28, 2008.

\bibitem{MN08}
V.~M{\"a}kinen and G.~Navarro.
\newblock Dynamic entropy-compressed sequences and full-text indexes.
\newblock {\em ACM Transactions on Algorithms}, 4(3):article 32, 2008.

\bibitem{Man01}
G.~Manzini.
\newblock {An Analysis of the Burrows-Wheeler Transform}.
\newblock {\em Journal of the ACM}, 48(3):407--430, 2001.

\bibitem{Munro86}
J.~I. Munro.
\newblock An implicit data structure supporting insertion, deletion, and search
  in {$O(\log n)$} time.
\newblock {\em Journal of Computer System Sciences}, 33(1):66--74, 1986.

\bibitem{Munro96}
J.~I. Munro.
\newblock {Tables}.
\newblock In {\em {Proc. 16th Foundations of Software Technology and Computer
  Science (FSTTCS)}}, LNCS 1180, pages 37--42, 1996.

\bibitem{MunRam01}
J.~I. Munro and V.~Raman.
\newblock Succinct representation of balanced parentheses and static trees.
\newblock {\em SIAM Journal on Computing}, 31(3):762--776, 2001.

\bibitem{MunRamRao01}
J.~I. Munro, V.~Raman, and S.~S. Rao.
\newblock Space efficient suffix trees.
\newblock {\em Journal of Algorithms}, 39(2):205--222, 2001.

\bibitem{MunroRamanStormSODA01}
J.~I. Munro, V.~Raman, and A.~J. Storm.
\newblock Representing dynamic binary trees succinctly.
\newblock In {\em {Proc. ACM-SIAM SODA}}, pages 529--536, 2001.

\bibitem{MunRao04}
J.~I. Munro and S.~S. Rao.
\newblock Succinct representations of functions.
\newblock In {\em Proc. 31th International Colloquium on Automata, Languages
  and Programming (ICALP)}, LNCS 3142, pages 1006--1015, 2004.

\bibitem{OS09}
D.~Okanohara and K.~Sadakane.
\newblock A linear-time burrows-wheeler transform using induced sorting.
\newblock In {\em Proc. 16th International Symposium on String Processing and
  Information Retrieval (SPIRE)}, LNCS 5721, pages 90--101, 2009.

\bibitem{Pag01}
R.~Pagh.
\newblock {Low Redundancy in Static Dictionaries with Constant Query Time}.
\newblock {\em {SIAM Journal on Computing}}, 31(2):353--363, 2001.

\bibitem{PV10}
M.~P{\v a}tra{\c s}cu and E.~Viola.
\newblock Cell-probe lower bounds for succinct partial sums.
\newblock In {\em Proc. 21st ACM-SIAM Symposium on Discrete Algorithms (SODA)},
  pages 117--122, 2010.

\bibitem{PD06}
Mihai P{\v a}tra{\c s}cu and Erik~D. Demaine.
\newblock Logarithmic lower bounds in the cell-probe model.
\newblock {\em SIAM Journal on Computing}, 35(4):932--963, 2006.

\bibitem{PT06}
M.~P\u{a}tra\c{s}cu and M.~Thorup.
\newblock Time-space trade-offs for predecessor search.
\newblock In {\em Proc. 38th Annual ACM Symposium on Theory of Computing
  (STOC)}, pages 232--240, 2006.

\bibitem{Pat08}
M.~P\v{a}tra\c{s}cu.
\newblock Succincter.
\newblock In {\em Proc. 49th IEEE Annual Symposium on Foundations of Computer
  Science (FOCS)}, pages 305--313, 2008.

\bibitem{RRR01}
R.~Raman, V.~Raman, and S.~S. Rao.
\newblock Succinct dynamic data structures.
\newblock In {\em Proc. 7th Annual Workshop on Algorithms and Data Structures
  (WADS)}, LNCS 2125, pages 426--437, 2001.

\bibitem{RRR02}
R.~Raman, V.~Raman, and S.~S. Rao.
\newblock Succinct indexable dictionaries with applications to encoding $k$-ary
  trees and multisets.
\newblock In {\em Proc. 13th Annual ACM-SIAM Symposium on Discrete Algorithms
  (SODA)}, pages 233--242, 2002.

\bibitem{RamRao03}
Rajeev Raman and S.~Srinivasa Rao.
\newblock Succinct dynamic dictionaries and trees.
\newblock In {\em Annual International Colloquium on Automata, Languages and
  Programming (ICALP)}, volume 2719 of {\em Lecture Notes in Computer Science},
  pages 357--368. Springer-Verlag, 2003.

\bibitem{Sada02a}
K.~Sadakane.
\newblock {Succinct Representations of {\it lcp} Information and Improvements
  in the Compressed Suffix Arrays}.
\newblock In {\em Proc. 13th Annual ACM-SIAM Symposium on Discrete Algorithms
  (SODA)}, pages 225--232, 2002.

\bibitem{Sada07a}
K.~Sadakane.
\newblock {Compressed Suffix Trees with Full Functionality}.
\newblock {\em {Theory of Computing Systems}}, 41(4):589--607, 2007.

\bibitem{Sad06a}
K.~Sadakane.
\newblock Succinct data structures for flexible text retrieval systems.
\newblock {\em {Journal of Discrete Algorithms}}, 5:12--22, 2007.

\bibitem{SNsoda10}
K.~Sadakane and G.~Navarro.
\newblock Fully-functional succinct trees.
\newblock In {\em Proc. 21st Annual ACM-SIAM Symposium on Discrete Algorithms
  (SODA)}, pages 134--149, 2010.

\bibitem{Vui80}
J.~Vuillemin.
\newblock A unifying look at data structures.
\newblock {\em Communications of the ACM}, 23(4):229--239, 1980.

\end{thebibliography}

\end{document}